\documentclass[10pt]{article}
\usepackage{array}

\usepackage{amsmath, amsthm, amssymb, graphics}
\usepackage[dvips]{graphicx}
\usepackage{enumerate, color, url}
\usepackage{natbib}
\usepackage{xr}
\newtheorem{theorem}{Theorem}
\newtheorem{lemma}{Lemma} 
\newtheorem{remark}{Remark}
\newtheorem{assumption}{Assumption}

\newtheorem{corollary}{Corollary}
\newtheorem{proposition}{Proposition}
\newtheorem{condition}{Condition}
\newtheorem{innercustomcnd}{Condition}
\newenvironment{customcnd}[1]
  {\renewcommand\theinnercustomcnd{#1}\innercustomcnd}
  {\endinnercustomcnd}

\newcommand{\bx}{\textbf{x}}

\newcommand{\aMST}{\texttt{aMST}}
\newcommand{\uMST}{\texttt{uMST}}
\newcommand{\aMDP}{\texttt{aMDP}}
\newcommand{\uNNG}{\texttt{uNNG}}
\newcommand{\CuMST}{\texttt{C-uMST}}
\newcommand{\CuNNG}{\texttt{C-uNNG}}
\newcommand{\Mstar}{\mathcal{M}_0^*}
\newcommand{\calT}{\mathcal{T}}
\newcommand{\calC}{\mathcal{C}}

\newcommand{\calE}{\mathcal{E}}
\newcommand{\calO}{\mathcal{O}}
\newcommand{\calA}{\mathcal{A}}

\newcommand{\bfP}{\mathbf{P}}
\newcommand{\bfE}{\mathbf{E}}

\newcommand{\bfPP}{\bfP_\texttt{P}}
\newcommand{\bfEP}{\bfE_\texttt{P}}

\newcommand{\VP}{\text{\textbf{Var}}_\texttt{P}}

\newcommand{\RC}{R_{ C_0}}

\graphicspath{{image/}}

\begin{document}

%%%%%%%%%%%%%%%%%%%%%%%%%%%%%%%%%%%%%%%%%%%%%%%%%%%%%%%%%%%%%%%%%%%%%%%%%%%%%%%%%%%%%%%%%%%%%%%%%%%%%%%%%%%%%%%%%%%%%%%%%%%%
%%%%%%%%%%%%%%%%%%%%%%%%%%%%%%%%%%%%%%%%%%%%%%%%%%%%%%%%%%%%%%%%%%%%%%%%%%%%%%%%%%%%%%%%%%%%%%%%%%%%%%%%%%%%%%%%%%%%%%%%%%%%

\renewcommand{\baselinestretch}{1.2}

%% commented for submitting to arXiv
% \markright{
% \hbox{\footnotesize\rm Statistica Sinica (2012): Preprint}\hfill
% }

\markboth{\hfill{\footnotesize\rm Hao Chen AND Nancy Zhang} \hfill}
{\hfill {\footnotesize\rm Graph-Based Tests for Categorical Data} \hfill}

\renewcommand{\thefootnote}{}
$\ $\par

%%%%%%%%%%%%%%%%%%%%%%%%%%%%%%%%%%%%%%%%%%%%%%%%%%%%%%%%%%%%%%%%%%%%%%%%%%%%%%%%%%%%%%%%%%%%%%%%%%%%%%%%%%%%%%%%%%%%%%%%%%%%

\fontsize{10.95}{14pt plus.8pt minus .6pt}\selectfont
\vspace{0.8pc}
\centerline{\large\bf Graph-Based Tests for Two-Sample Comparisons}
\centerline{\large \bf of Categorical Data}
%\vspace{2pt}
%\centerline{\large\bf IF A SECOND LINE IS NEEDED}
\vspace{.4cm}
\centerline{Hao Chen}
\vspace{.4cm}
\centerline{\it Department of Statistics, Stanford University}
\vspace{.4cm}
\centerline{Nancy R. Zhang}
\vspace{.4cm}
\centerline{\it Department of Statistics, The Wharton School, University of Pennsylvania}
\vspace{.55cm}
\fontsize{9}{11.5pt plus.8pt minus .6pt}\selectfont

%%%%%%%%%%%%%%%%%%%%%%%%%%%%%%%%%%%%%%%%%%%%%%%%%%%%%%%%%%%%%%%%%%%%%%%%%%%%%%%%%%%%%%%%%%%%%%%%%%%%%%%%%%%%%%%%%%%%%%%%%%%%

\begin{quotation}
\noindent {\it Abstract:}

\par

We study the problem of two-sample comparison with categorical data when the contingency table is sparsely populated.  In modern applications, the number of categories is often comparable to the sample size, causing existing methods to have low power.  When the number of categories is large, there is often underlying structure on the sample space that can be exploited.   We propose a general non-parametric approach that utilizes similarity information on the space of all categories in two sample tests.   Our approach extends the graph-based tests of \cite{friedman1979multivariate} and \cite{rosenbaum2005exact}, which are tests base on graphs connecting observations by similarity.  Both tests require uniqueness of the underlying graph and cannot be directly applied on categorical data.  We explored different ways to extend graph-based tests to the categorical setting and found two types of statistics that are both powerful and fast to compute.  We showed that their permutation null distributions are asymptotically normal and that their $p$-value approximations under typical settings are quite accurate, facilitating the application of the new approach.  The approach is illustrated through several examples.

%Two types of statistics that are powerfully efficient and fast to compute are filtered out.  Permutation distributions of both statistics are studied and shown to be asymptotically normal as the number of observed categories goes to infinity.  The accuracy of the $p$-value obtained through the normal approximation is checked and shown to be quite good when the number of categories not too small.  The application of this new approach to different types of problems is illustrated via examples.
\vspace{9pt}

\par
\end{quotation}\par

%%%%%%%%%%%%%%%%%%%%%%%%%%%%%%%%%%%%%%%%%%%%%%%%%%%%%%%%%%%%%%%%%%%%%%%%%%%%%%%%%%%%%%%%%%%%%%%%%%%%%%%%%%%%%%%%%%%%%%%%%%%%

\fontsize{10.95}{14pt plus.8pt minus .6pt}\selectfont

\section{Introduction}
\label{sec:intro}
    Testing whether two data samples are drawn from the same distribution is a fundamental problem in statistics.  For low-dimensional Euclidean data, there are many approaches, both parametric and non-parametric, to this problem.  When the data are categorical, the existing approaches are much more limited.  The standard procedure is to assume that each sample is drawn from a multinomial distribution, and the comparison becomes a test of whether the two samples come from the same multinomial distribution.  Classical methods, such as the Pearson's Chi-square test and the deviance test, work well when we observe each category a large number of times.  At least, the region in the contingency table where the two groups truly differ needs to be adequately sampled for existing tests to achieve good power.  However, in many modern applications, the number of possible categories is comparable to or even larger than the sample size.  Some examples are the following.

    \begin{description}
        \item[Preference rankings:]  Survey data in marketing or psychometric research often come in the form of preference rankings.   Subjects may be asked to rate wine (rank from best to worst tasting), pictures (choose 3 most familiar out of 5), or insurance plans (identify the most and least desirable).  See \cite{diaconis1988group} and \cite{critchlow1985metric} for more detailed examples on ranked and partially ranked data.  It is a common problem to compare two groups of subjects to see if there is any between-group difference in preference.  The number of possible full rankings is the factorial of the number of objects being rated, and the number of possible rankings is higher if some subjects only partially rank the objects.

        \item[Haplotype association:]  In genetics, a haplotype is a combination of alleles at adjacent loci on a chromosome that is transmitted together. A common problem of genetic association studies is to compare haplotype counts between treatment and control groups (e.g. see \cite{zaykin2002testing} and Furihata, Ito and Kamatani \citeyearpar{furihata2006test}). Each haplotype can be represented as a fixed-length binary vector.  The number of possible haplotypes is exponential in the number of loci.  Haplotypes that are longer than 10 are often of interest in genetics, leading to $>1,000$ possible combinations.  However, the number of subjects in association studies is often only in the thousands or even hundreds, and the counts for most haplotypes are small.

        \item[Sequence or document comparisons:]  In the modern age of digitized texts, it is often of interest to compare the word composition in two different documents.  A similar problem is the comparison of DNA or protein sequences, which plays a large role in bioinformatics (Lippert, Huang and Waterman \citeyearpar{lippert2002distributional}).  The number of possible words in these applications can be very large, while the counts for most words are small.  For recent interest in this problem see \citet{PerryBeiko}, \citet{BushLahn} and Rajan, Aravamuthan and Mande \citeyearpar{Rajanetal} for examples.
    \end{description}

Classical Chi-square tests have low power in these scenarios due to sparsity of the contingency table and high dimensionality of the parameter space.  %Even if we increase the sample size to 10-fold, there are so many possible categories that many of the new subjects would explore those uncovered categories that the resulting contingency table would still be sparsely populated though with more categories.
% The generalized Fisher's exact test for multiple categories would also not work since the number of extreme scenarios would be too large.
For exact tests, it is possible to generalize the concept to the setting of more than two categories, but this is computationally challenging (\cite{mehta1983network}) and not efficient due to the existence in high dimensions of many tables that have the same probability as the one observed.

When the number of categories is very large,  there is often underlying similarity between different categories that can be exploited.  For example, rankings can be related through Kendall's or Spearman's distance.  Hamming distance or other more sophisticated measures can be used to compare haplotypes and fixed-length words in DNA sequences.  In document comparison, the similarities between words are not equally likely: Some words are synonyms of others; Some are more likely to be used together.  Such similarity information between categories can be used to improve the power of two-sample tests.  %After defining the similarity measure on the categories, a distance matrix on all subjects can be obtained and a graph on the subjects can be constructed with edges connecting similar subjects.  From the underlying rational, if the two samples come from the same distribution, the number of edges in the graph connecting subjects from different groups would be large.

We assume that a distance matrix has been given on the set of categories, and adopt the graph-based approach proposed by \cite{friedman1979multivariate} and \cite{rosenbaum2005exact}, where a graph is constructed on all subjects so that subjects more similar in value are connected by an edge.  Friedman and Rafsky's test is based on a minimum spanning tree (MST), and Rosenbaum's test is based on minimum distance pairing (MDP).  The test statistic in both cases is the number of edges connecting subjects from different groups.  The underlying rationale is that, if two groups come from the same distribution, subjects coming from the same group should be as distant to each other as subjects coming from different groups.  %If each category has only one observation and there is no tie in the distance matrix, the approaches proposed by through constructing a minimum spanning tree (MST), and by \cite{rosenbaum2005exact} through constructing a minimum distance pairing (MDP), and the one based on constructing a nearest neighbor graph (NNG), can be directly applied.
More details of these tests are given in Section \ref{sec:graphreview}.  Both tests, however, require uniqueness of the underlying graphs.  When the distance matrix on subjects is filled with ties, which is characteristic of categorical data, neither approach can be directly applied.

Ties in the distance matrix lead to ambiguity in constructing the MST or MDP, and the number of possible graphs increases rapidly with the number of ties. Some efforts were made to address this problem.  In the analysis of a partially ranked data set with 38 subjects in 23 categories, \cite{critchlow1985metric} tried both the graph obtained from the union of all MSTs (uMST), and the graph obtained from the union of all nearest neighbor graphs (uNNG).  \cite{nettleton2001testing} also used uNNG on a binary clinical feature data set with 64 subjects in 63 categories.  In general, nearest neighbor graphs do not work well for categorical data, see Section \ref{sec:stat}.  In this paper, Critchlow's method using the uMST is studied in more detail and a computationally tractable form for categorical data is given.  A different statistic, based on averaging over all optimal graphs of a certain kind, is also proposed and analyzed.

%This problem is avoided in Nettleton and Banerjee's work since only one category contains more than one subject for that specific data set

%Though \cite{critchlow1985metric} proposed to use uMST as the graph, the approach was only applied on a specific data set with sample size 38 and no universally computable form of the statistic was given.
In Section \ref{sec:stat}, analytically tractable forms of the two statistics based on averaging over and union of minimum spanning trees are derived and compared via simulation to statistics based on MDP and NNG.  %Two statistics, one based on averaging over all MSTs and the other based on uMST, out stand themselves from other statistics in simulation studies.
While the two MST-based tests are shown to be more powerful than the MDP- and NNG-based tests, neither the averaged nor the union-based statistic dominate in power for the simulation scenarios explored.
Algorithmic details for computing these two statistics are described and, in particular, the averaged statistic is shown to be computationally intractable for some problems.  A generalized version of the averaged statistic, with better computational properties, is proposed.  In Section \ref{sec:examples}, the graph-based approach is illustrated in simulations and data examples, and shown to have much better power than Chi-square tests.  In Section \ref{sec:null}, permutation null distributions of the proposed statistics are described.  After mean- and variance- standardization, the statistics are shown to be asymptotically normal, under certain assumptions on the cell counts and the graph's structure, as the number of observed categories goes to infinity.  %The normal approximation allows for fast making the theoretical $p$-value calculation instant.  %The normal approximation to $p$-value is checked through numerical studies under different settings and shown to be pretty good when the number of observed categories not too small, which is good enough for real applications since otherwise calculating $p$-value through permutation is acceptable.

\section{Preliminaries}
\subsection{Notations}

We start by introducing our notation.  The different
categories are indexed by $1, 2, \dots, K$, with arbitrary naming of the categories.  The two groups are labeled $a$ and $b$, and the data are given in the form of a two-way contingency table (Table \ref{tab:cont}).  Without loss of generality, we assume that each category has at least one subject over the two groups.  That is, categories with no observation in either group can be omitted from the analysis without loss of information.

\begin{table}[h] \centering
 \caption{Basic Notations.}\label{tab:cont}

\bigskip

  \begin{tabular}{|c|c|c|c|c|c|} \hline & 1 & 2 & $\dots$ & K & Total \\ \hline \hline
Group $a$ & $n_{a1}$ & $n_{a2}$ & $\dots$ & $n_{aK}$ & $n_a$ \\ \hline
Group $b$ & $n_{b1}$ & $n_{b2}$ & $\dots$ & $n_{bK}$ & $n_b$ \\ \hline
Total & $m_1$ & $m_2$ & $\dots$ & $m_K$ & $N$ \\ \hline
    \end{tabular}
$$ m_k  = n_{ak} + n_{bk},\ k = 1,\dots, K;$$
$ n_a  =\sum_{k=1}^K n_{ak},\quad n_b=\sum_{k=1}^K n_{bk}, \quad N  = n_a + n_b = \sum_{k=1}^K m_k. $
% \begin{align*}
% m_k & = n_{ak} + n_{bk},\ k = 1,\dots, K;\\
% n_a & =\sum_{k=1}^K n_{ak},\quad n_b=\sum_{k=1}^K n_{bk}, \quad N  = n_a + n_b = \sum_{k=1}^K m_k.
% \end{align*}

\end{table}
We sometimes refer to individual subjects themselves and denote them by $Y_1,\dots,Y_N$.  Thus, each $Y_i$ takes value in $\{1,\dots,K\}$ and has a group label
\begin{equation}
g_i = \left\{
        \begin{array}{ll}
          a, & \hbox{if $Y_i$ belongs to group $a$;} \\
          b, & \hbox{if $Y_i$ belongs to group $b$.}
        \end{array}
      \right.
\end{equation}
We assume that a distance matrix, $\{d(i,j):~i,j=1,\dots,K\}$ has been given on the set of possible categories, with $d(i,j)$ small if categories $i$ and $j$ are similar.  Possible ways of defining the distance matrix are shown for various examples in Section \ref{sec:intro}.

A graph $G$ is defined by its vertices and edges.  We use $G$ to refer to both the graph and its set of edges when the vertex set is implicitly obvious.  $|\cdot|$ is used to denote the size of the set, so $|G|$ is the number of edges in $G$.  For any node $i$ in the graph $G$, $\mathcal{E}_i^G$ denotes the set of edges in $G$ that contain node $i$, $\mathcal{V}_i^G$ denotes the set of nodes in $G$ that are connected to node $i$ by an edge, and $\mathcal{E}_{i,2}^G$ denotes the set of edges in $G$ that contain at least one node in $\mathcal{V}_i^G$.  For any event $A$, $I_A$ is the indicator function that takes value 1 if $A$ is true and 0 otherwise.

% Following is a list of abbreviations for different types of graphs and test statistics:
% \begin{description}
% \item MST: Minimum Spanning Tree,
% \item MDP: Minimum Distance Pairing,
% \item NNG: Nearest Neighbor Graph,
% \item uMST: The graph obtained by taking the union of all MSTs,
% \item uNNG: The graph obtained by taking the union of all NNGs, equivalent to the graph connecting each point to all of its nearest neighbors,
% \item $R_G$: The test statistic on the graph $G$,
% \item $R_\aMST$: the test statistic averaging over all test statistics computed on each of the MSTs,
% \item $R_\aMDP$: the test statistic averaging over all test statistics computed on each of the MDPs,
% % \item $R_\uMST$: the test statistic computed on uMST,
% % \item $R_\uNNG$: the test statistic computed on uNNG.
% \end{description}

\subsection{A Review of Graph-Based Two-Sample Tests} \label{sec:graphreview}
By \emph{graph-based two-sample tests}, we refer to tests that are based on graphs with the subjects $\{Y_i\}$ as nodes.  We here suppose $\{Y_i\}$ take distinct values such that certain graphs can be constructed uniquely.  The graph can be constructed using the distance matrix on $\{Y_i\}$.  \cite{friedman1979multivariate} proposed the first graph-based two-sample test as a generalization of the Wald-Wolfowitz runs test to multivariate settings.
Their test is based on a MST on the subjects, which is a spanning tree connecting all subjects that minimizes the sum of distances across edges.  The Friedman-Rafsky test is based on the number of edges connecting subjects across different groups:
\begin{equation} \label{sumedges}
    \sum_{(i,j) \in G} I_{g_i \neq g_j},
\end{equation}
where $G$ is the MST.  The statistic is standardized to have mean zero and variance one, and its value is compared to the null distribution obtained by permuting the group labels.  Friedman and Rafsky showed that, while this test has low power in low dimensions,  it has comparable power to likelihood ratio tests in a numerical study of normal data in $>20$ dimensions, and higher power when the normal assumption is violated.

Another graph-based two-sample method, the cross-match test, was proposed by \cite{rosenbaum2005exact}.  This test is based on a minimum distance non-bipartite pairing (MDP) that divides the $N$ subjects into $N/2$ (assuming $N$ is even) non-overlapping pairs in such a way as to minimize the total of $N/2$ distances between pairs.  For odd $N$ Rosenbaum suggested creating a pseudo data point that has distance 0 with all other subjects, and later discarding the pair containing this pseudo point.  The sum (\ref{sumedges}) is computed with $G$ set to the MDP.  The test statistic is the mean- and variance- standardized version of this sum.  Note that the topology of the MDP does not depend on the distance matrix, with each node always having degree 1.  This fact makes the test based on MDP truly distribution-free under the null hypothesis.

Both methods assume uniqueness of the type of graph used.  For categorical data, ties appear in the distance matrix whenever a category has multiple counts.  Even sparse contingency tables have quite a few cells containing more than one subject.   The number of possible graphs grows rapidly with the number of ties.  Thus, Friedman and Rafsky's and Rosenbaum's methods cannot be directly applied to categorical data.
For categorical data, distances are often based on qualitative measures and thus, while their relative ranking may be trustworthy, their absolute scale is not.  Hence, we do not consider methods based directly on the distance matrix.  While there are many ways to construct a graph based on a distance matrix, we limit our study to MST, MDP, and NNG as representative.  Figure \ref{fig:graph1} illustrates the three different types of graphs on a simple example containing six points. These six points take on six distinct values.

\begin{figure} \centering
  \includegraphics[width=.6\textwidth]{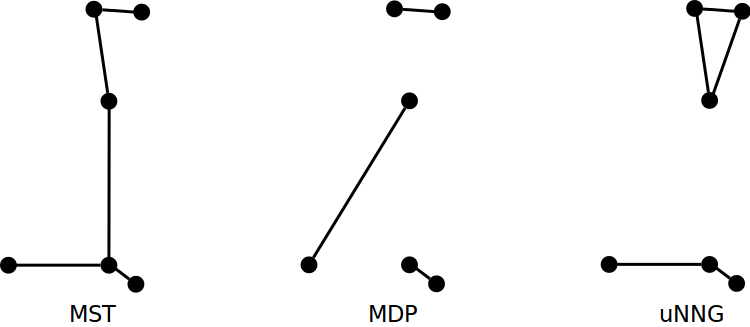}
  \caption{Illustration of MST, MDP, and NNG on six points.  Notice that only one of the two possible MSTs on the six points and one of the two possible NNGs on the six points are shown.}
  \label{fig:graph1}
\end{figure}

\section{Generalized Graph-Based Test Statistics} \label{sec:stat}
One natural solution, when the optimizing graph is not unique, is to average the test statistic over all graphs of the given kind.   In this section, we consider the statistic based on averaging \eqref{sumedges} over all MSTs ($R_\aMST$).  Another solution to non-uniqueness it to take the union over all optimizing graphs, such as the statistic based on the uMST ($R_\uMST$).  $R_\aMST$ and $R_\uMST$ are analytically tractable and intuitively appealing, and their derivations are shown in Section \ref{sec:mststat}.  For comparison, we also consider the statistic based on averaging \eqref{sumedges} over all MDPs, $R_\aMDP$, and the statistic based on uNNG, $R_\uNNG$.  Computation of $R_\aMDP$, described in Appendix \ref{sec:Rmdp}, is often intractable.  Computation of uNNG is instantaneous.   In Section \ref{sec:sim}, we study by simulation the performance of $R_\aMST, R_\uMST, R_\aMDP,$ and $R_\uNNG$, comparing them to each other and to Chi-square tests.  Our results show that tests based on minimum spanning trees have the best power, and the intuition for this is explained.   The statistics based on uMDP and average over all NNGs are not included in the comparison because they do not have the potential of high power according to the performance of $R_\aMDP$ and $R_\uNNG$ in Section \ref{sec:sim}, and calculating them is not instant.  To clarify ambiguities, $R_G$ is used to denote the test statistic on graph $G$ in general, with exceptions for $R_\aMST$ and $R_\aMDP$.

 Computation of $R_\aMST$ and $R_\uMST$ is described in more detail in Section \ref{sec:comp}.  When the number of MSTs on \emph{categories} is large, which is common for categorical data, computation for $R_\aMST$ can be very costly.  We generalize the statistic based on $R_\aMST$ to a similar but simpler form in Section \ref{sec:general}.

    \subsection{The Test Statistics Based on MST} \label{sec:mststat}
%We here derive the computationally tractable forms for $R_\aMST$ and $R_\uMST$.
\subsubsection{$R_\aMST$} \label{sec:RaMST}
For each $k=1,\dots,K$, let $\calC_k \subset \{1,\dots,N\}$ be the subjects that belong to category $k$.  From Table \ref{tab:cont}, $|\calC_k|=m_k$.  Let $\calT_k$ be the set of all spanning trees for $\calC_k$.  Since the distance between any two subjects in $\calC_k$ is zero, any spanning tree of $\calC_k$ is a MST of $\calC_k$.  Let $\calT^*_0$ be the set of all MSTs on the categories. We can embed each tree in $\calT^*_0$ as a graph on the subjects by randomly picking one subject in $\calC_k$ to represent category $k$, for $k=1,\dots,K$. For each $\tau_0^* \in \calT^*_0$, there are $\prod_{k=1}^K m_k^{|\mathcal{E}_k^{\tau_0^*}|} $
different embeddings.  For example, Figure \ref{fig:mst_real} shows 3 out of 15552 ($=2\cdot 3^3 \cdot 1 \cdot 4^2 \cdot 3^2 \cdot 2$) possible embeddings for a MST on six categories containing 2, 3, 1, 4, 3 and 2 subjects.  Let $\calT_0$ be the set of all graphs obtained from embedding a tree from $\calT^*_0$ on the subjects. Then $|\calT_0| =  \sum_{\tau_0^*\in\calT_0^*}
\left(\prod_{k=1}^K m_k^{|\mathcal{E}_k^{\tau_0^*}|} \right)$.
\begin{figure} \centering
  \includegraphics[width=.6\textwidth]{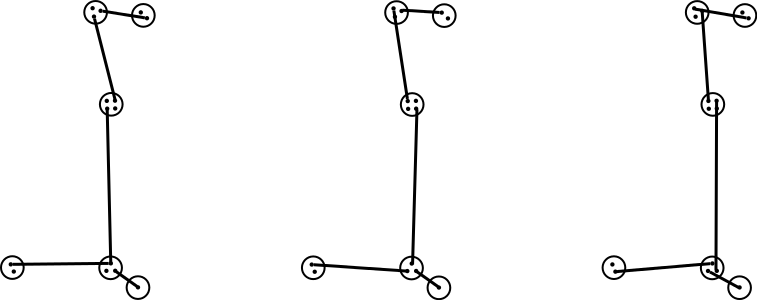}
  \caption{Embedding the MST on categories on the subjects.  This
    figure only shows 3 out of 15552 possible embeddings.}
  \label{fig:mst_real}
\end{figure}
Let $\calT$ be the set of all MSTs on the $N$ subjects.  Then, any member of $\calT$ can be represented as a union of a graph from $\calT_0$ and a graph from each of $\{\calT_k:~k=1,\dots,K\}$, and vice versa.  Thus,
$$ \calT = \left\{ \tau_0 \cup (\bigcup_{k=1}^K \tau_k): \tau_0\in \calT_0, \tau_k\in\calT_k, k=1, \dots, K \right\},$$
with $|\calT| =|\calT_0| \prod_{k=1}^K S_{m_k}$,
where $S_m=m^{m-2}$ is the number of spanning trees on $m$ points, by Cayley's formula.  % For any tree $\tau \in \calT$, let $R_{\tau}$ be the sum (\ref{sumedges}) computed on $\tau$.
The test statistic based on averaging all MSTs on subjects is
\begin{equation}
R_\aMST \overset{\Delta}{=}  |\calT|^{-1} \sum_{\tau\in \calT}R_\tau,
\end{equation}
where $R_\tau$ is (\ref{sumedges}) with $G=\tau$.  One can give a computationally tractable form for $R_\aMST$ in terms of the cell counts of the contingency table and the set of possible MSTs on categories.
\begin{theorem}
  \label{thm:Rmst} The test statistic based on averaging over all MSTs on subjects is
\begin{align}
R_\aMST & = \sum_{k=1}^K \frac{2n_{ak}n_{bk}}{m_k} + |\calT_0|^{-1} \sum_{\tau_0^*\in\calT_0^*}
\prod_{k=1}^K m_k^{|\mathcal{E}_k^{\tau_0^*}|}
\sum_{(u,v) \in \tau_0^*}\frac{n_{au} n_{bv} + n_{av}
n_{bu}}{m_u m_v}. \label{eq:Rmst}
\end{align}
\end{theorem}
The proof for Theorem \ref{thm:Rmst} is in Supplementary material \ref{sec:proof-theorem}.

 The statistic $R_\aMST$ has a much simpler form if there is a unique MST on categories, or if the total number of subjects in each category is the same.
\begin{corollary}\label{thm:Rmst_sp1}
  When $|\calT_0^*|=1$, then
\begin{equation}
  \label{eq:Rmst_sp1} R_\aMST = \sum_{k=1}^K
\frac{2n_{ak}n_{bk}}{m_k} + \sum_{(u,v) \in \tau_0^*}
\frac{n_{au} n_{bv} + n_{av} n_{bu}}{m_u m_v},
\end{equation} where $\tau_0^*$ is the unique MST on categories.
\end{corollary}

% \begin{corollary}\label{thm:Rmst_sp2}
% When $m_k \equiv m,\ k=1,\dots,K$,
% \begin{equation}
%   \label{eq:Rmst_sp2} R_\aMST = \sum_{k=1}^K
% \frac{2n_{ak}n_{bk}}{m} + |\calT_0^*|^{-1} \sum_{\tau_0^*\in\calT_0^*}
% \sum_{(u,v) \in \tau_0^*}\frac{n_{au} n_{bv} + n_{av}
% n_{bu}}{m^2}.
% \end{equation}
% \end{corollary}

The form  (\ref{eq:Rmst_sp1}) of the statistic is intuitive.  For each category $k$, we call the term $2n_{ak}n_{bk}/m_k$ the \emph{mixing potential} of the category.  The mixing potential is maximized when the subjects in category $k$ are evenly divided between groups $a$ and $b$; it is minimized when the category contains subjects from only one group.  A mixing potential for each edge $(u,v)$ can also be defined as $(n_{au}n_{bv} + n_{av}n_{bu})/(m_um_v)$.  The edge-wise mixing potential is maximized when the edge connects a category containing only group $a$ subjects with a category containing only group $b$ subjects; it is minimized when both categories contain subjects only from one group.  Thus, mixing potentials over categories and over edges between categories measure the similarity between the two groups.  Corollary \ref{thm:Rmst_sp1} shows that, when the MST on categories is unique, the test statistic $R_\aMST$ reduces to the sum of mixing potentials over nodes and edges of the MST on categories.  The similarity information on the categories is explicitly incorporated into the test through the sum of mixing potentials over the edges between categories.  % For the general form of $R_\aMST$, the mixing potentials over the edges are weighted according to their occurrences in MSTs.
In testing, (\ref{eq:Rmst}) and (\ref{eq:Rmst_sp1}) %and (\ref{eq:Rmst_sp2})
must be compared to their permutation distributions.  A generalized statistic proposed later in Section \ref{sec:general} is based directly on (\ref{eq:Rmst_sp1}).

\subsubsection{$R_\uMST$} \label{sec:RuMST}
%Other than $R_\uNNG$, which can be calculated immediately since obtaining uNNG is fast, calculating $R_\uMST$ needs some extra work.  Obtaining uMST can be troublesome when there are multiple MSTs on categories.
Let $\Mstar$ denote the set of edges appearing in at least one MST on categories,
$$\Mstar = \{ (u,v)\in \tau_0^*: \tau_0^* \in \mathcal{T}_0^* \}.$$
Thus $\Mstar$ is the uMST with the categories as nodes.  When there is only one MST on categories, $\tau_0^*$, then $\Mstar=\tau_0^*$; when there are multiple MSTs on categories, which is common for categorical data, obtaining $\Mstar$ is not straightforward.  Computation of $\Mstar$ is discussed in Section \ref{sec:comp}.  One can state the analytic form of $R_\uMST$ given $\Mstar$.
\begin{theorem}
  The test statistic based on uMST is
  \begin{equation}
    \label{eq:RuMST}
    R_\uMST = \sum_{k=1}^K n_{ak}n_{bk} + \sum_{(u,v)\in \mathcal{M}_0^*} (n_{au}n_{bv} + n_{av}n_{bu}),
  \end{equation}
% where $\mathcal{M}_0^*$ is the union of all $\tau_0 \in \mathcal{T}_0$.
\end{theorem}
\begin{proof}
  Within each category, every pair of subjects is connected, which gives the first term of \eqref{eq:RuMST}.  If categories $u$ and $v$ are connected in any $\tau_0^* \in \mathcal{T}_0^*$, then each point in category $u$ is connected to every point in category $v$, giving the second term of \eqref{eq:RuMST}.  If categories $u$ and $v$ are not connected in any $\tau_0^* \in \mathcal{T}_0^*$, no edge will appear between categories $u$ and $v$ in uMST.
\end{proof}

\begin{remark}\label{remark:RuMST}
\emph{Both $R_\uMST$ and $R_\aMST$ are derived from sums of $I_{g_i \neq g_j}$ over edges of the uMST on subjects.  The main difference between them is that $R_\uMST$ treats all of the edges equally, while $R_\aMST$ assigns  each edge a weight proportional to the number of MSTs on subjects in which the edge appears. Comparing \eqref{eq:RuMST} to \eqref{eq:Rmst_sp1},} %(assuming the simplest case that there is only one MST on categories, $\Mstar=\tau_0^*$) 
\emph{the denominators in \eqref{eq:Rmst_sp1} are omitted in \eqref{eq:RuMST}.   Each edge within category $k$ appears in $|\mathcal{T}|/(m_k/2)$ MSTs, while each edge between categories appears in $|\mathcal{T}|/(m_u m_v)$ MSTs. Therefore, in comparison with $R_\aMST$,  $R_\uMST$ puts relatively more weight on between-category edges than within-category edges.}  %It is hard to say which way is better now and the performances of $R_\aMST$ and $R_\uMST$ are compared in simulation studies.
\end{remark}

    \subsection{A Numerical Study} \label{sec:sim}
The power of the tests based on $R_\aMST$, $R_\uMST$, $R_\aMDP$ and $R_\uNNG$ was studied and compared to Pearson's Chi-square and deviance tests on simulated data sets.  In each simulation, 30 points were randomly sampled from different distributions -- $N(0,1)$ vs $N(1,1)$, $N(0,1)$ vs $N(0,4)$, $N(0,1)$ vs $N(1,4)$, and $U(0,5)$ vs $U(1,6)$.  The combined sample of 60 points was then discretized into 12 bins of equal width.  The value 12 was chosen so that the average number of data points per category was 5, mimicking the low cell count scenario. %(In the first three simulations on data from Normal distributions, cells in the middle range would have observations enough for Chi-square tests, while other cells not.)
The bins were ranked by their start positions, and the distance between two categories was defined as the difference in their ranks.  The $p$-values for all tests were calculated through 1,000 permutation samples for each simulation run, and the power was obtained from 1,000 simulation runs.  Figure \ref{fig:sim} shows power versus type I error for each test and each simulation setting, and a table listing the power under 0.05 significance level.  %We only show type I error up to 0.10 since higher type I errors are not of interest and we can see the performances of the tests much clearer in current scale.
 In the plots, since Pearson's Chi-square and deviance tests gave similar results, only the results for the deviance test are shown.  The deviance test is denoted by ``LR'' since it is based on the log-likelihood ratio.

First, compare $R_\aMST$, $R_\aMDP$, and $R_\uNNG$.  $R_\aMST$ was always significantly more powerful than $R_\aMDP$, which in turn was always more powerful than $R_\uNNG$. This result is intuitive from the definition of the different graphs.  Since the MST must span the entire data set, $K-1$ out of its $N-1$ edges are forced to connect points between categories.  For MDP, if a category has even number of subjects, the subjects in that category would be paired amongst themselves; between-category edges are only possible for those categories having an odd number of subjects.  For uNNG, as long as a category has more than one subject, the subjects in that category would not be connected to subjects from other categories.  Therefore, tests based on MST make the most use of the similarity information among categories, while the test based on $R_\uNNG$ makes the least use of this information.  The simulation results show a positive correlation between using similarity information and the power of the test.

In simulations, $R_\uMST$ and $R_\aMST$ performed similarly under the scenarios that compared two Normal distributions, while $R_\uMST$ had very little power, even lower than $R_\aMDP$ and the deviance test, for the comparison of two Uniform distributions with different supports.  When comparing two Normal distributions, the similarity between two categories was closely related to the difference of the ranks of the categories.  That is, the further apart the ranks of the two categories, the less similar.  However, when comparing two Uniform distributions with different supports, only the ranks at the two ends are informative while the middle ranks are not. Since $R_\uMST$ puts more weight on between-category edges compared to $R_\aMST$, its power would be lower if the similarity measure among categories were not informative.  Note that of all the graph-based tests, only the test based on $R_\aMST$ consistently outperformed the deviance test.

    \begin{figure}[!htp]
      \centering
      \includegraphics[width=.49\textwidth]{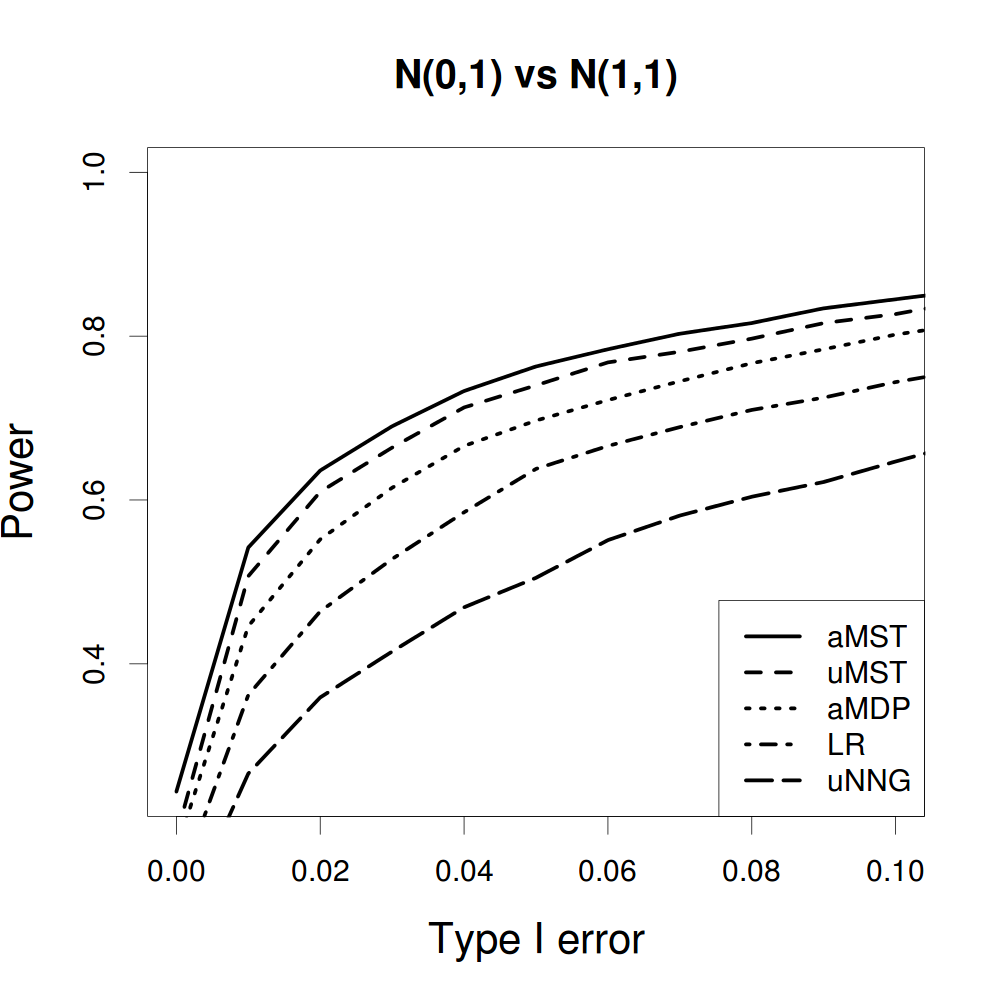}
      \includegraphics[width=.49\textwidth]{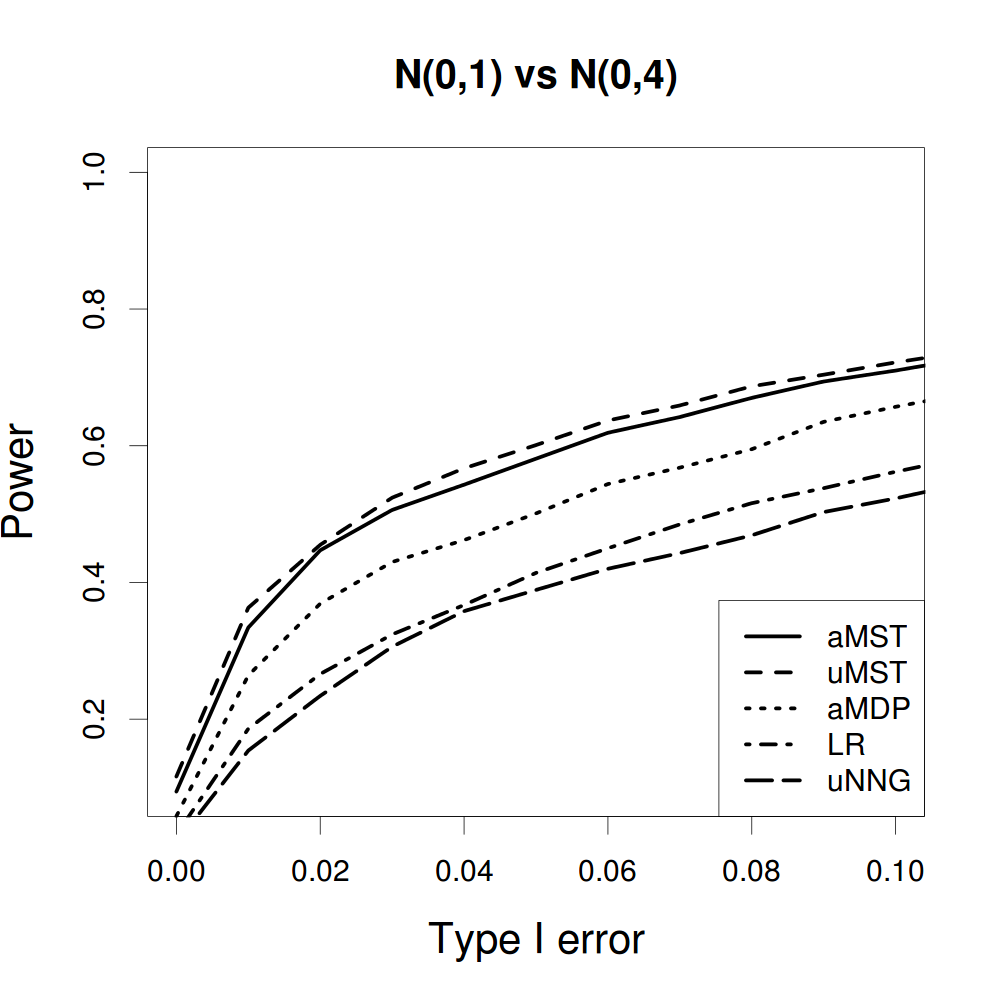} \\
      \includegraphics[width=.49\textwidth]{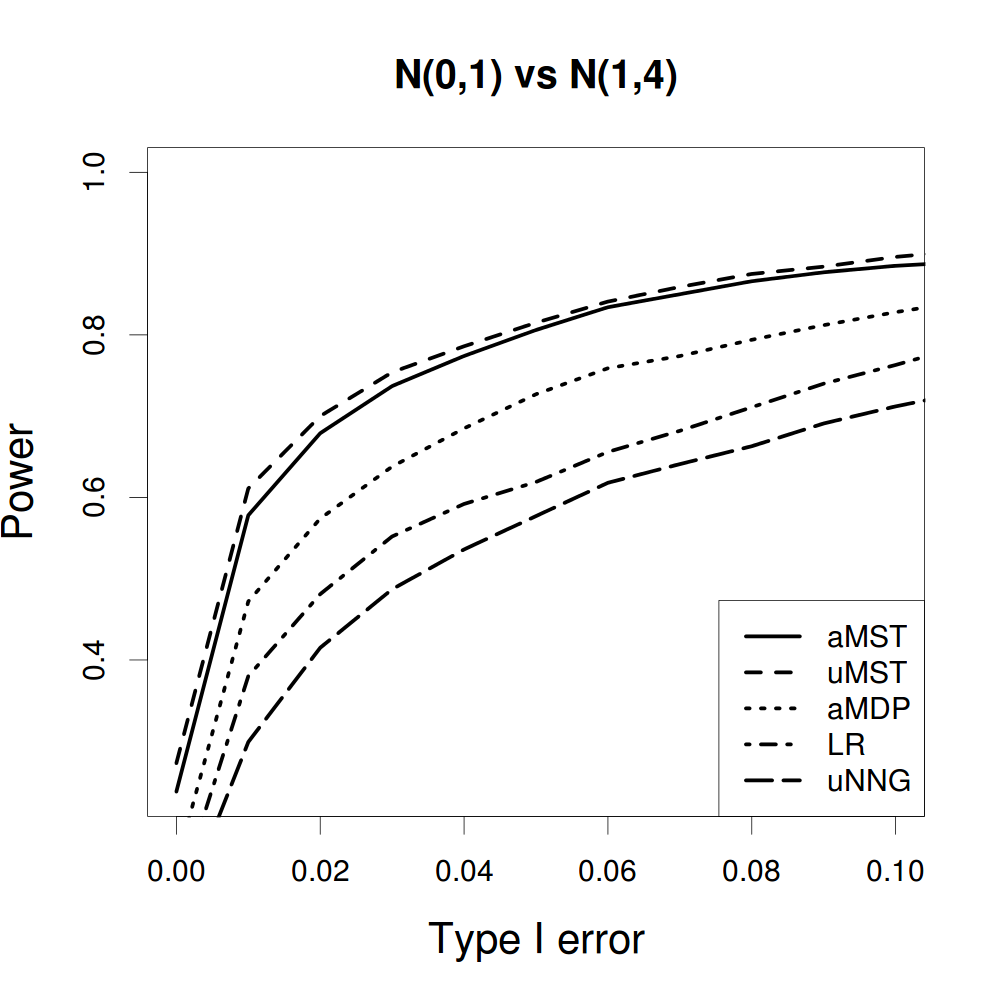}
      \includegraphics[width=.49\textwidth]{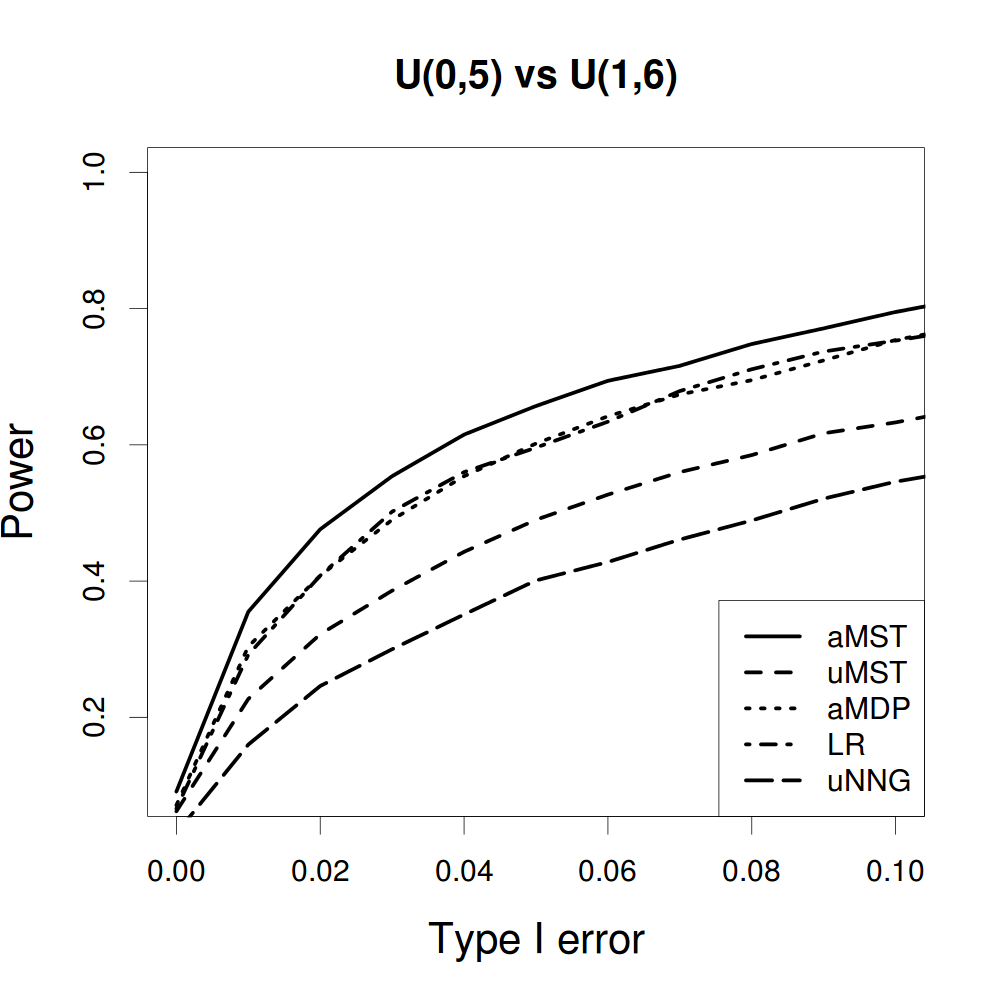} \\
\ \\
      \begin{tabular}{c|cccccc}
 \hline
 & aMST & uMST & aMDP & uNNG & LR & Pearson \\ \hline
N(0,1) vs N(1,1)  & 0.762 & 0.740 & 0.679 & 0.492 & 0.605 & 0.605 \\ \hline
N(0,1) vs N(0,4) & 0.558 & 0.585 & 0.482 & 0.382 & 0.394 & 0.396 \\ \hline
N(0,1) vs N(1,4) & 0.804 & 0.824 & 0.722 & 0.569 & 0.632 & 0.626 \\ \hline
U(0,5) vs U(1,6) & 0.665 & 0.486 & 0.607 & 0.383 & 0.600 & 0.552 \\ \hline
      \end{tabular}

      \caption{Power versus type I error for tests based on $R_\aMST$, $R_\uMST$, $R_\aMDP$, the likelihood ratio (deviance), and $R_\uNNG$ under different simulation settings.  The table lists the power under 0.05 significance level.}
      \label{fig:sim}
    \end{figure}

%     \begin{table}[!htp]
%       \centering
%       \begin{tabular}{|c|cccccc|}
% \hline \hline
%  & aMST & uMST & aMDP & uNNG & LR & Pearson \\ \hline \hline
% N(0,1) vs N(1,1)  & 0.762 & 0.740 & 0.679 & 0.492 & 0.605 & 0.605 \\ \hline
% N(0,1) vs N(0,4) & 0.558 & 0.585 & 0.482 & 0.382 & 0.394 & 0.396 \\ \hline
% N(0,1) vs N(1,4) & 0.804 & 0.824 & 0.722 & 0.569 & 0.632 & 0.626 \\ \hline
% U(0,5) vs U(1,6) & 0.665 & 0.486 & 0.607 & 0.383 & 0.600 & 0.552 \\ \hline \hline
%       \end{tabular}
%       \caption{The power of six tests -- four graph-based tests based on $R_\aMST$, $R_\uMST$, $R_\aMDP$, $R_\uNNG$, the deviance test (LR) and Pearson's Chi-square test -- under 0.05 significance level and different simulation settings.}
%       \label{tab:sim}
%     \end{table}

This simulation study is limited and only used ranked data.  We chose this study design for its interpretability.  Though simple, the results are informative and show the advantage of averaged MST over averaged MDP and uNNG for categorical data.  Also, averaged MST is better than uMST when the similarity measure used to construct the graph is not effective;  if the similarity measure is effective, the test based on uMST is comparable to, and sometimes better than, the test based on averaged MST.  Hence we focus on the tests based on $R_\aMST$ and $R_\uMST$.

    \subsection{Computational Issues} \label{sec:comp}
The analytic forms \eqref{eq:Rmst} and \eqref{eq:RuMST}, require enumeration of all MSTs on categories for $R_\aMST$; and enumeration of all edges in $\mathcal{M}_0^*$ for $R_\uMST$.  %These tasks are unavoidable because \eqref{eq:Rmst} and \eqref{eq:RuMST} cannot be further simplified in general.
Let $M = |\calT_0^*|$ be the number of MSTs on categories.
If the distance matrix between categories is continuous-valued, then usually $M=1$.  Even when the distance matrix is arithmetic, $M$ is often small enough to be manageable.  However, for problems that exhibit certain symmetries, enumeration of the set of all MSTs on categories is not computationally feasible.  For the haplotype association problem in Section \ref{sec:hap}, the number of MSTs on categories $M$ can be computed using the Matrix-Tree Theorem if we assume all categories are non-empty: $$ M =  2^{2^l-l-1} \prod_{i=2}^l \exp \left\{ \scriptsize{\left(  \begin{array}{c} l \\ i \end{array}\right) \log i }  \right\},$$
where $l$ is the haplotype length.  From the formula for $M$, it increases quickly as $l$ increases.  When the length of the haplotype is 6, a reasonably short length in genetic studies, there are only 64 possible categories while $M$ equals $1.66\times 10^{45}$.  One may argue that in this case, \eqref{eq:Rmst} may be further simplified using the symmetry over categories, so that enumeration of $|\calT_0^*|$ is not necessary.  This is true if all categories are non-empty, but if one or more of the categories are empty, the symmetry breaks and $M$ is still too large for enumeration.

 %   \begin{table}[t]
%     \centering
%       \label{tab:Nmst}
%       \begin{tabular}{|c|c|c|} \hline Haplotype Length ($l$) & $K$ &  $M$\\ \hline \hline
% 2 & 4 & 4\\ \hline
% 3 & 8 & 384 \\ \hline
% 4 & 16 & 42467328 \\ \hline
% 5 & 32 & $2.078\times 10^{19}$\\ \hline
% 6 & 64 & $1.66\times 10^{45}$ \\ \hline
%       \end{tabular}

%       \caption{ The number of categories, $K$, and the number of MSTs on categories, $M$, as haplotype length increases for the haplotype association problem in Section \ref{sec:hap}. All categories are assumed to be non-empty.}
%     \end{table}

 %Table \ref{tab:comptime} summarizes the computation time for $R_\aMST$ and $R_\uMST$ in terms of $K$ and $M$.
Consider the listing of all edges in uMST on categories, $\Mstar$, which is required for $R_\uMST$.  This task can be completed in $\mathcal{O}(K^2)$ time through an algorithm proposed by \cite{eppstein1995representing}.  Details of the algorithm are in Appendix \ref{sec:comptime}, and its theoretical justification is completed by \cite{chen2013graph}.  The $\mathcal{O}(K^2)$ time is usually affordable since $K$ is no larger than the sample size.  Thus $R_\uMST$ is computationally feasible for any problem.  On the other hand, $R_\aMST$ requires the enumeration of all MSTs on categories, not just their edges, and thus adds $\mathcal{O}(M)$ computation time to the algorithm.  For the haplotype example, this makes $R_\aMST$ computationally infeasible.  In the next Section, we propose a statistic that is motivated by $R_\aMST$ but is computationally tractable for all problems.

% \begin{table}[!htp]
%   \centering
%   \begin{tabular}{|c|c|c|}
% \hline
%     & Task & Computation Time \\ \hline
% $R_\aMST$ & Enumerating all MSTs on categories & $\mathcal{O}(K^2+M)$ \\ \hline
% $R_\uMST$ & Listing edges in uMST on categories & $\mathcal{O}(K^2)$ \\ \hline
%   \end{tabular}
%   \caption{Computational time for $R_\aMST$ and $R_\uMST$. $M$ is the number of MSTs on categories.}
%   \label{tab:comptime}
% \end{table}

\subsection{A Fast Method Generalized from $R_\aMST$} \label{sec:general}
Corollary \ref{thm:Rmst_sp1} gives a simple and intuitive form of $R_\aMST$ when there is a unique MST on categories.  In that special case, $R_\aMST$ is the sum of mixing potentials computed within each category and mixing potentials computed between categories that are connected by an edge of the MST $\tau_0^*$.  Evidence against the null increases if this sum of mixing potentials is small, as compared to random permutation.    In (\ref{eq:Rmst_sp1}), the MST $\tau_0^*$ serves as an enumeration of the pairs of categories that are highly similar.  There is nothing sacred about the choice of MST for this role.  The intuitive interpretation for (\ref{eq:Rmst_sp1}) remains if we replace $\tau_0^*$ by any other graph $ C_0$ that represents proximity between categories.

We assumed so far that a distance matrix on categories is used to represent the similarity between categories.  We now bypass the distance matrix and assume that similarity is directly represented by a graph $ C_0$ with the categories as nodes.  Our goal is to incorporate the proximity information encoded by the graph into the two group comparison.  We propose a statistic obtained by substituting $ C_0$ for $\tau_0^*$ in (\ref{eq:Rmst_sp1}):
\begin{equation}
  \label{eq:R_G0} \RC = \sum_{k=1}^K \frac{2n_{ak}n_{bk}}{m_k} +
\sum_{(u,v) \in  C_0}\frac{n_{au} n_{bv} + n_{av}
n_{bu}}{m_u m_v}.
\end{equation}
There is a similar interpretation as for $R_\aMST$.  Consider all $C_0$-spanning graphs that are graphs on subjects, where every pair of subjects are connected by a path if they are in the same category or they are in two categories that are connected by a path in $C_0$.  Hence,
% graphs that span the subjects (i.e. connect every subject with every other subject via a path).  A $ C_0$-spanning graph is a graph spanning all subjects that includes, for every edge $(u,v)$ in $ C_0$, an edge that connects a subject from $u$ and a subject from $v$.  A minimum distance $ C_0$-spanning graph is a $ C_0$-spanning graph that minimizes the sum of distances on its edges.  It is easy to see that
minimum distance $ C_0$-spanning graphs connect subjects within categories by spanning trees, and connects exactly one pair of subjects between each pair of categories that have an edge in $ C_0$.  $\RC$ is the averaged sum (\ref{sumedges}) over all minimum distance $ C_0$-spanning graphs.

If $ C_0$ is given, computing $\RC$ requires $\mathcal{O}(K+| C_0|)$ time.  If $ C_0$ is not given, the choice of $ C_0$ can often be guided by domain knowledge.  In our examples, choices for $C_0$ include the uMST on categories that we denote by C-uMST (same as $\Mstar$), and the uNNG on categories that we denote by C-uNNG.   Since C-uMST and C-uNNG can both be computed in $\mathcal{O}(K^2)$ time, $R_\CuMST$ and $R_\CuNNG$ require only $O(K^2)$ computation time for any problem.

\section{Examples} \label{sec:examples}

The application of $R_\CuMST$, $R_\CuNNG$ and $R_\uMST$ are illustrated on several examples, both real and simulated.  In the simulated examples, their powers are compared to those of Chi-square tests.  The $p$-values for all tests were calculated through 1,000 permutation samples for each run, and the power calculated through 1,000 simulation runs.  %All three graph-based tests have power much higher then the Chi-square tests.  Among the three, $R_\uMST$ has the highest power, while $R_\CuMST$ and $R_\CuNNG$ are similar with $R_\CuMST$ slightly better, which is probably due to the fact that C-uNNG is a subset of C-uMST.

    \subsection{Preference Ranking} \label{sec:ranking} Consider comparing two groups of subjects on the ranking of four objects.  Let $\Xi$ be the set of all permutations of the set $\{1,2,3,4\}$.  Data were simulated as follows:  Subjects from group $a$ have no preference among the four objects, and their rankings were uniformly drawn from $\Xi$;  rankings of subjects from group $b$ were generated from the distribution
\begin{equation}
    P_{\theta}(\zeta) = \frac{1}{\psi(\theta)}\exp\{-\theta d(\zeta,\zeta_0)\}, \quad \zeta, \zeta_0 \in \Xi, ~\theta \in \mathbb{R},
\end{equation}
where $d(\cdot,\cdot)$ is a distance function and $\psi$ a normalizing constant.  This probability model, first considered by \cite{mallows1957non} with Kendall's or Spearman's distance, favors rankings that are similar to a modal ranking $\zeta_0$ if $\theta>0$.  See \cite{diaconis1988group} for more discussion.  The larger the value of $\theta$, the more clustering there should be in group $b$ around the mode $\zeta_0$.  We experimented with both Kendall's and Spearman's distance and various values for $\theta$.  We assumed that the true distance function used to generate the data is either known and used to construct the graph, or unknown, in which case an incorrect distance is used.

Figure \ref{fig:permgraph} shows C-uMST and C-uNNG formed on a typical data set generated under $\theta=5$ with $n_a = n_b = 20$.  Spearman's distance was used in both the generating model and for constructing the graph.  In this instance, C-uMST contained all edges in C-uNNG with three extra edges, shown in thinner lines.  The reason this happened is that no category was as close to category ``3241'' as category ``3142'', and no category is as close to category ``3142'' as category ``3241''.  For MST on categories, more edges are needed to form a spanning tree.  It is clear that in this case, there are three MSTs on categories, each one obtained by adding one of the three thinner edges to the C-uNNG.  In most simulation runs, C-uMST and C-uNNG were the same, while in those runs where they differed, C-uNNG was always a subset of C-uMST.

Figure \ref{fig:perm} shows the power versus type I error, for $\theta=5$ and $ n_a=n_b=20$, under different combinations of using Kendall's or Spearman's distance for the generating model and for constructing the graph, as well as the power under 0.05 significance level.  We see that even when a wrong distance was used, the graph-based tests still had significantly higher power than the Chi-square tests.  For this simulation setting, $R_\uMST$ was the most powerful among the three graph-based tests; $R_\CuMST$ and $R_\CuNNG$ performed similarly with $R_\CuMST$ a little better in all cases, implying that the extra edges in C-uMST gave additional useful information.

\begin{figure} \centering
  \includegraphics[width=.35\textwidth]{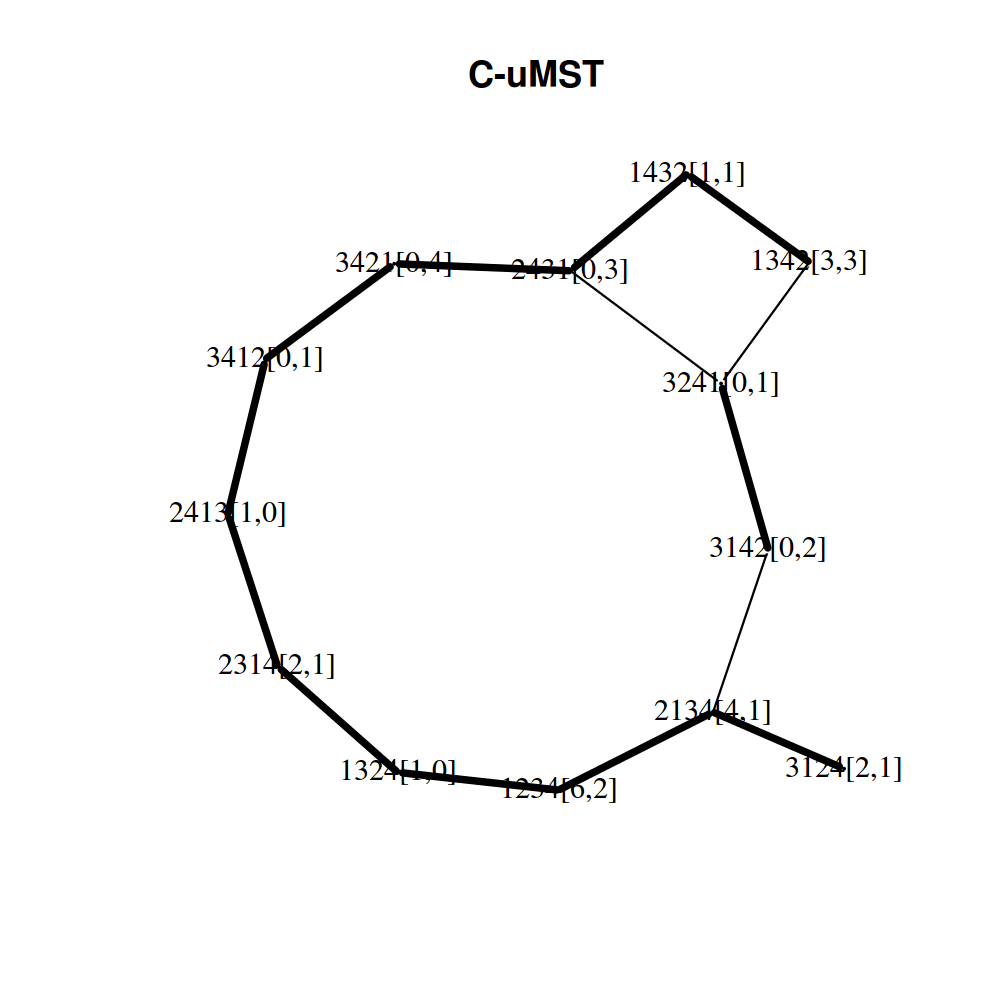} \quad
  \includegraphics[width=.35\textwidth]{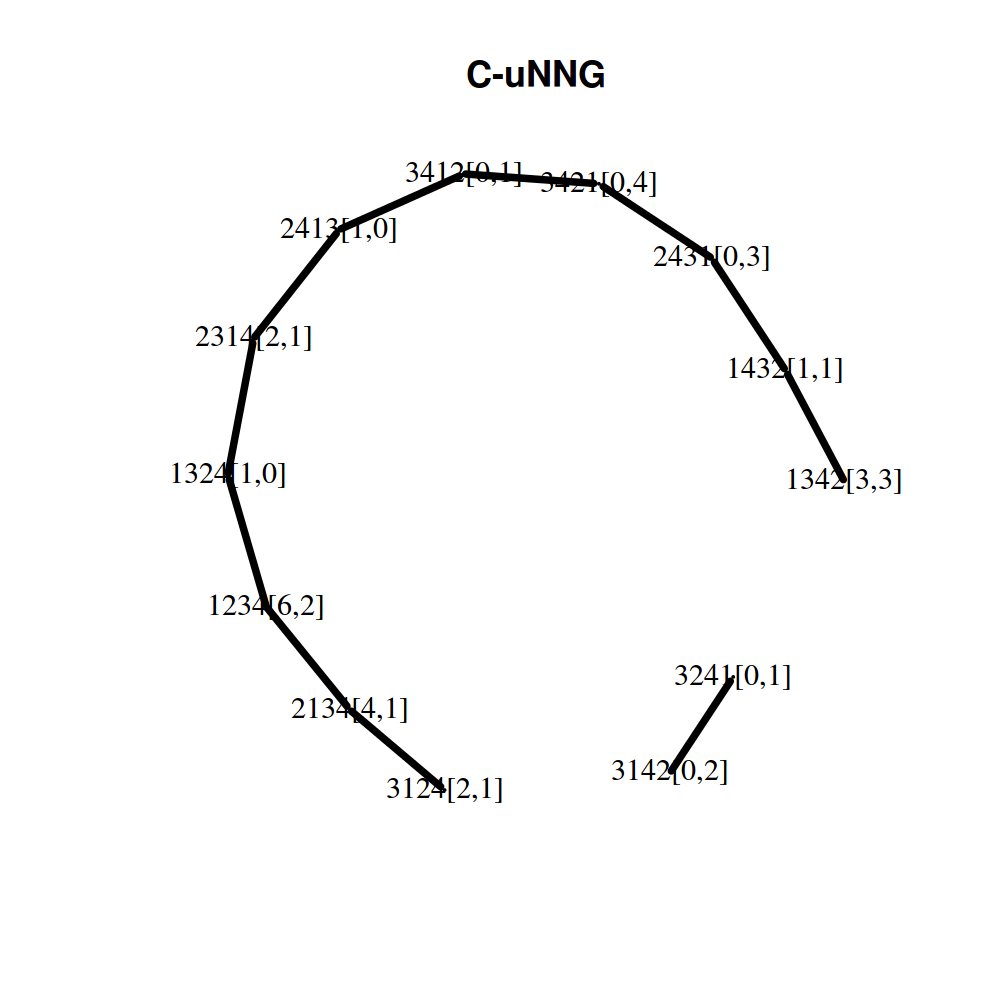}
  \caption{C-uMST and C-uNNG constructed on a typical data set generated under parameters $\zeta_0=1234$ and $\theta=5$, with $n_a=n_b=20$.  Spearman's distance was used in both the generating model and for constructing the graph.  Each node is labeled by the ranking it represents, followed by the number of subjects from groups $a$ and $b$  with that ranking in parentheses.}
  \label{fig:permgraph}
\end{figure}

\begin{figure}
  \centering
  \includegraphics[width=.45\textwidth]{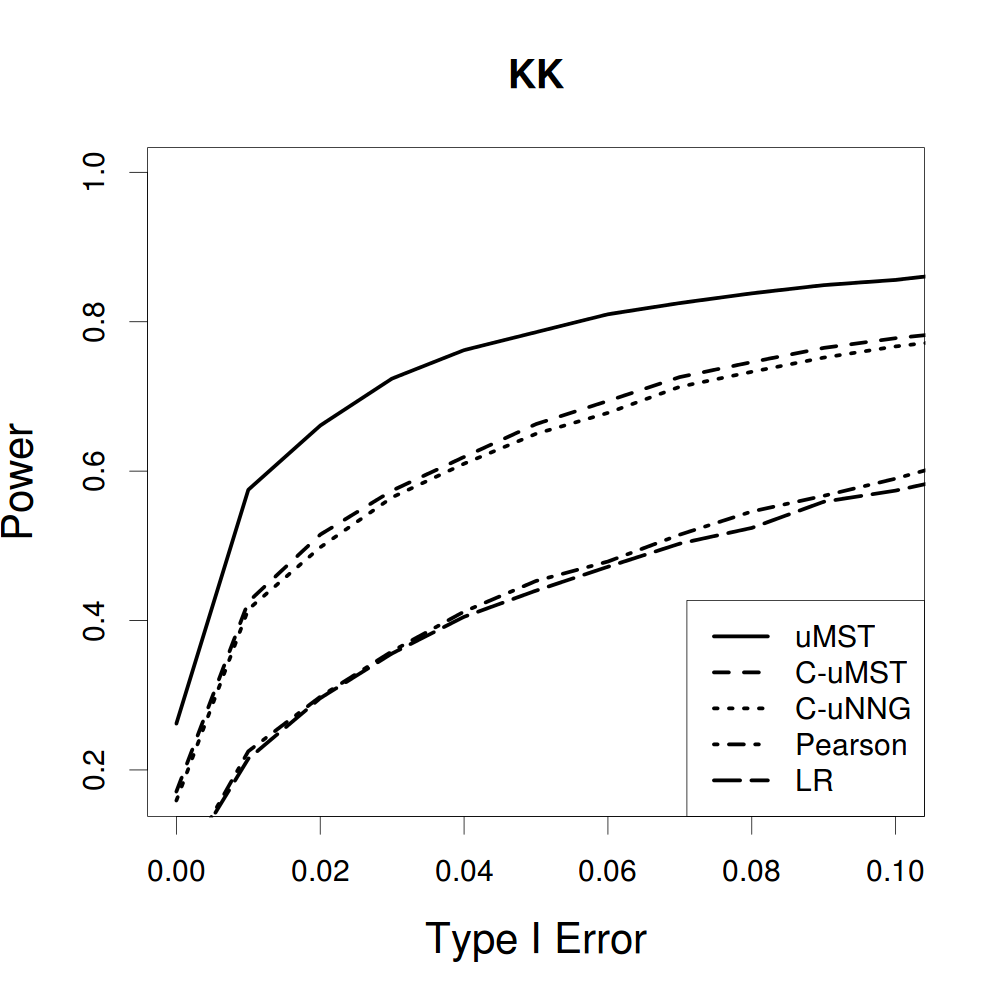}
  \includegraphics[width=.45\textwidth]{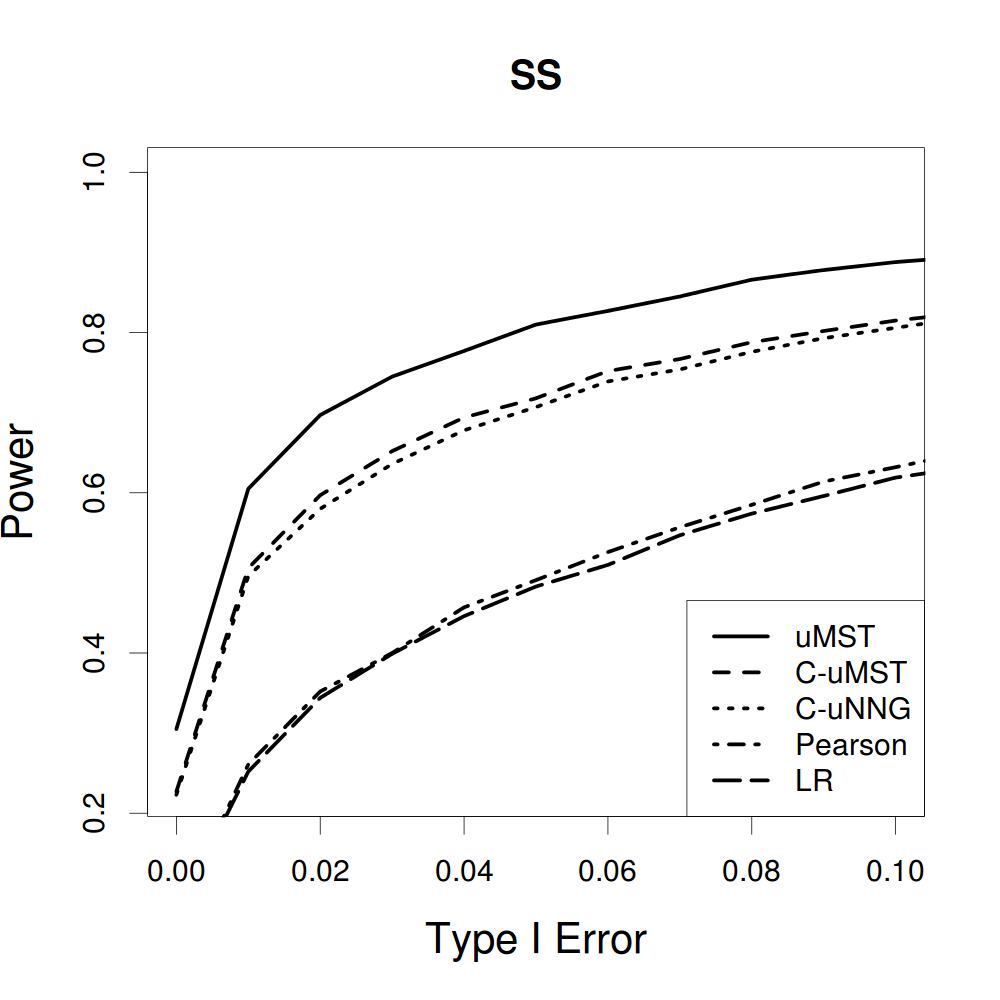} \\
  \includegraphics[width=.45\textwidth]{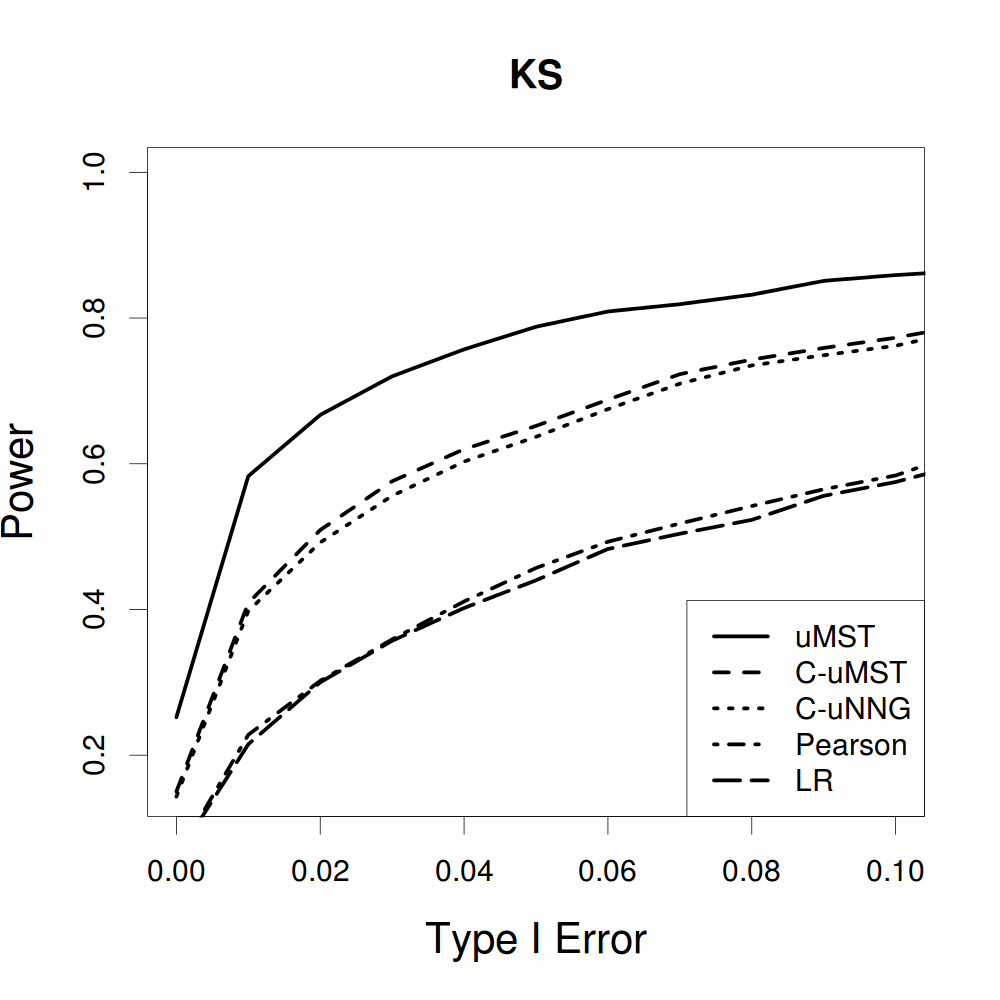}
  \includegraphics[width=.45\textwidth]{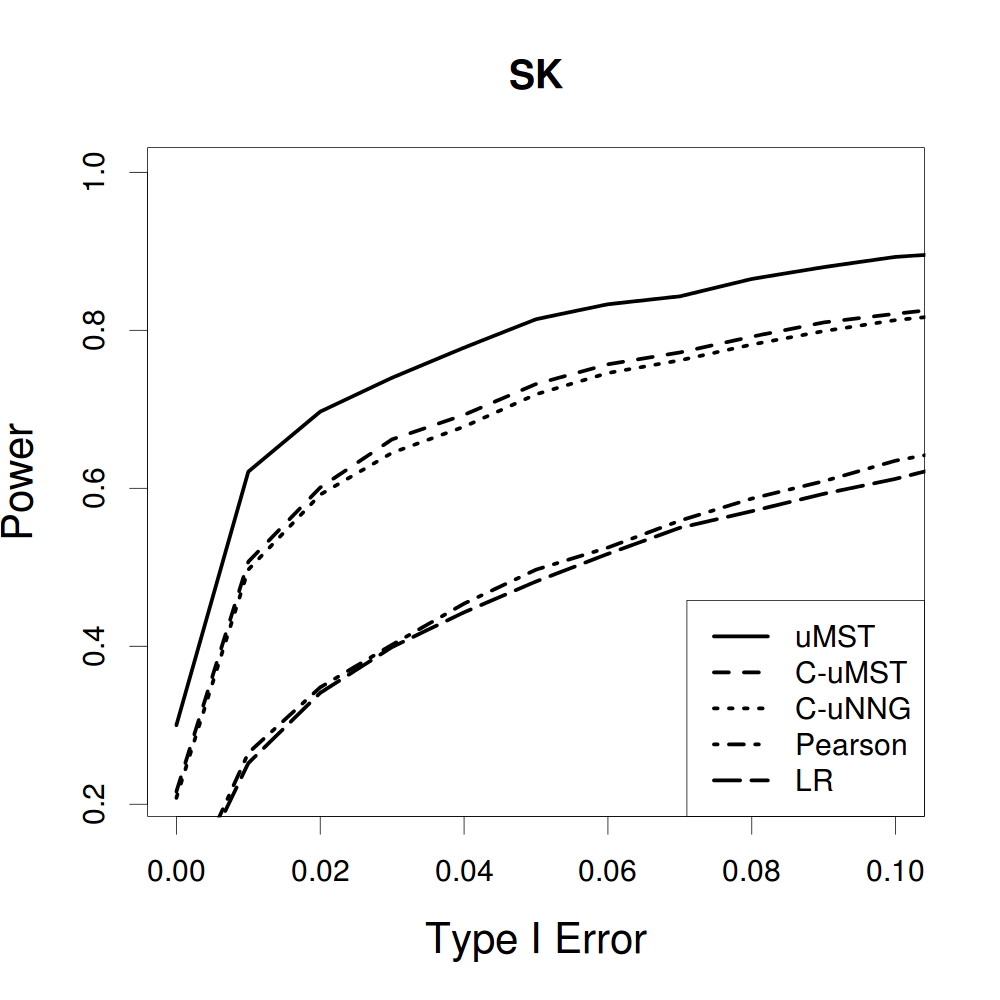} \\
\ \\
     \begin{tabular}{l|ccccc}
\hline
 & uMST & C-uMST & C-uNNG & Pearson & LR \\ \hline
KK & 0.784 & 0.660 & 0.648 & 0.450 & 0.439 \\ \hline
KS & 0.784 & 0.649 & 0.631 & 0.455 & 0.437 \\ \hline
SS & 0.807 & 0.715 & 0.703 & 0.485 & 0.480 \\ \hline
SK & 0.811 & 0.729 & 0.715 & 0.494 & 0.481 \\ \hline
      \end{tabular}

  \caption{Power versus type I error for the five tests in the preference ranking example with $\theta=5$ and $n_a=n_b=20$.   A distance measure, Kendall's (K)  or Spearman's (S) distance, was used for the generating model and for constructing the graph.  The first letter denotes the distance used in the generating model, and the second letter denotes the distance used in constructing the graph.  The table lists the powers under 0.05 significance level.}
  \label{fig:perm}
\end{figure}

% \begin{table}
%   \centering
%      \begin{tabular}{|l|ccccc|}
% \hline \hline
%  & uMST & C-uMST & C-uNNG & Pearson & LR \\ \hline \hline
% KK & 0.784 & 0.660 & 0.648 & 0.450 & 0.439 \\ \hline
% KS & 0.784 & 0.649 & 0.631 & 0.455 & 0.437 \\ \hline
% SS & 0.807 & 0.715 & 0.703 & 0.485 & 0.480 \\ \hline
% SK & 0.811 & 0.729 & 0.715 & 0.494 & 0.481 \\ \hline \hline
%       \end{tabular}
%       \caption{The power of five tests -- three graph-based tests based on $R_\uMST, R_\CuMST, R_\CuNNG$ and two Chi-square tests -- under 0.05 significance level and different simulation settings.}
%   \label{tab:perm}
% \end{table}

    \subsection{Haplotype Association}  \label{sec:hap} 
We consider a disease model where the probability for disease depends on the haplotype at four single nucleotide polymorphisms (SNP).  We encode the allele at each SNP as 0 or 1, and so the haplotype can be represented as a binary string.  We assume that the disease probability depends on the number of positions at which the subject's haplotype agrees with a target haplotype:
$$P(\hbox{Disease}) =  0.3 + 0.1 \times (\hbox{Number of positions in agreement}).$$
Thus, the probability of disease can take values 0.3 0.4, 0.5, 0.6 or 0.7 depending on whether there are 0, 1, 2, 3 or 4 positions in agreement.  To make the problem harder, we assume that seven non-informative SNPs are analyzed together with the four informative SNPs, and that which and how many of the 11 SNPs are informative is unknown in the analysis.  Thus the data actually consists of haplotypes of length 11.  There are $2^{11} = 2,048$ possible categories.  In each simulation, 1,000 haplotypes with length 11 were generated uniformly from all possible haplotypes.  Each subject with a given haplotype was signed as ``patient'' or ``normal'' according to the disease model.  Since only 1,000 subjects were simulated in each run, not all of the 2,048 categories were represented.  The number of non-empty categories in each run ranged from 755 to 823, with an average of 791 in the 1000 simulation runs.  The Hamming distance was used to construct the graph.  Figure \ref{fig:haplotype} shows the power versus type I error plots for the five tests.  It is clear that, by incorporating the information in the graph, tests based on $R_\uMST, R_\CuMST$, and $R_\CuNNG$ all have much higher power than the Pearson's Chi-square and deviance tests.  Among the three graph-based tests, the one based on $R_\uMST$ works a little better than the ones based on $R_\CuMST$ and $R_\CuNNG$.

\begin{figure} \centering
  \includegraphics[width=.5\textwidth]{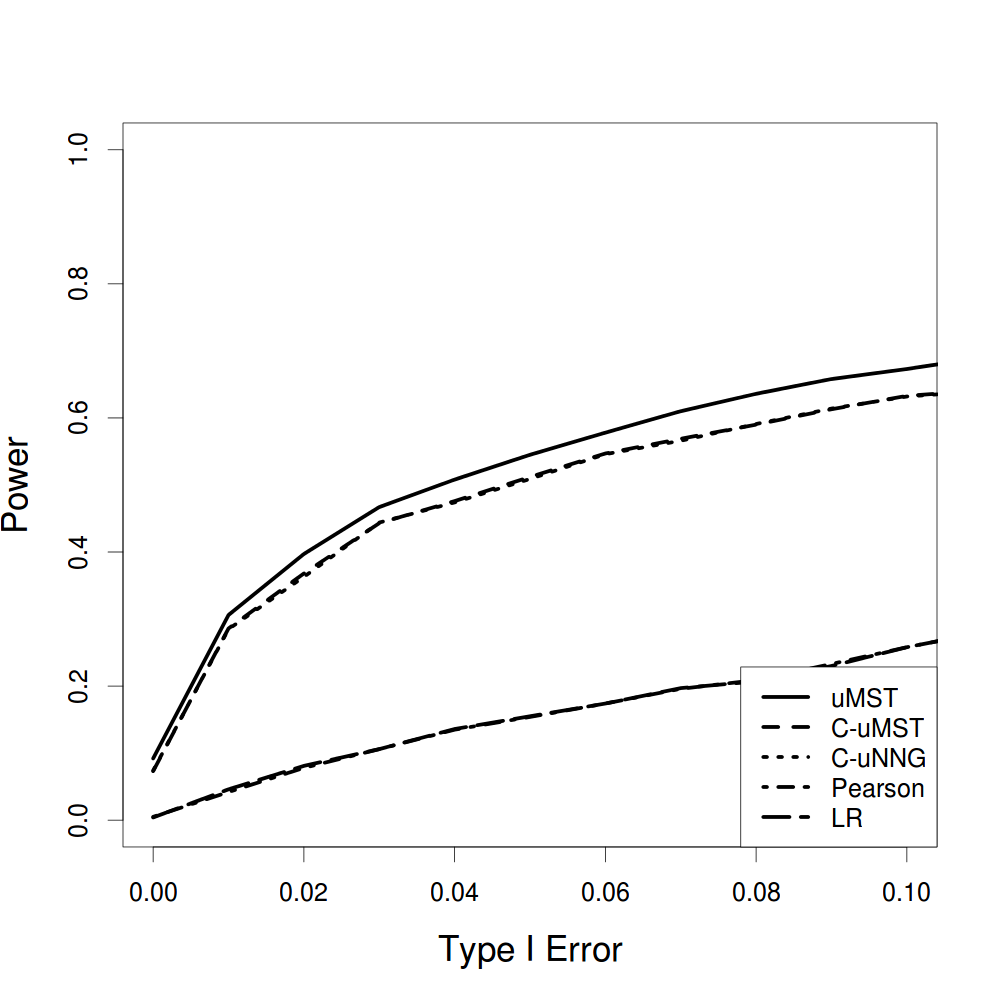}
  \caption{The power versus type I error plots for the five tests for the haplotype example.  The length of the haplotype is 11, with only 4 positions informative.}
  \label{fig:haplotype}
\end{figure}

    \subsection{Binary Clinical Features}  This example comes from \citet{anderson1972statistical} and \citet{nettleton2001testing}.  Data on the presence or absence of 17 clinical features of the eye ailment Keratoconjunctivitis Sicca (KCS) are given for two groups of patients.  A question asked by Nettleton and Banerjee was whether the two groups of patients share a common distribution with respect to these clinical features.  The sizes of the groups are 40 and 24. It turned out that only two subjects had the same outcome for the 17 clinical features, so there are in total 63 distinct categories.  Hamming distance was used to construct the graph, and $p$-values were calculated through $10,000$ permutation samples and are shown in Table \ref{tab:kcs}.  Nettleton and Banerjee's method is based on the uNNG on subjects.  As discussed before and confirmed by simulation studies in Section \ref{sec:sim}, the uNNG on subjects has lower power than MST based tests when many categories have more than one subject.  This is not a problem in this data set because only one category has more than one subject.  We see that $R_\uMST$, $R_\CuMST$, and $R_\CuNNG$ all detected the difference between the two groups of patients, while the Chi-square tests did not.

    \begin{table}[!htp]
      \centering
      \begin{tabular}{|c|c|c|c|c|c|}
        \hline
        $R_\uMST$ & $R_\CuMST$ & $R_\CuNNG$ & Nettleton and Banerjee's & Pearson & LR \\ \hline
        0.0011 & 0.0010 & 0.0006 & 0.0007 & 0.5200 & 0.5200 \\ \hline
      \end{tabular}
      \caption{$P$-values for the KCS data set.}
      \label{tab:kcs}
    \end{table}

\section{Permutation Distributions of the Test Statistics} \label{sec:null}
Based on the results in Sections \ref{sec:sim}-\ref{sec:general}, we focus now on $R_\CuMST$ and $R_\uMST$.  We consider the permutation distributions of these statistics in their generalized forms.  That is, we consider $R_{C_0}$ and $T_{C_0}$, the latter defined as
\begin{equation}
  \label{eq:TC0}
  T_{ C_0} = \sum_u n_{au}n_{bu} + \sum_{(u,v)\in  C_0} (n_{au}n_{bv} + n_{bu}n_{av})
\end{equation}
$T_\CuMST$ is equivalent to $R_\uMST$.  The permutation distributions of $R_\CuMST$ and $R_\uMST$ follow immediately.
%The more general forms are studied here because the same set of conditions are need in developing the results in this section, while in real applications, $C_0$ might be given, or we sometimes prefer to use other choices for $C_0$ rather than C-uMST.

We use two quantities to characterize the permutation distributions:
\begin{eqnarray}
\lambda &:=& \max_u|\calE_u^{C_0}|, ~~\hbox{the maximum node degree in $ C_0$.}\\
\beta &:=& \max_u m_u,~~ \hbox{the maximum total count for a category.}
\end{eqnarray}
By permutation distribution, we are referring to the distribution of the statistic under random uniform permutation of the group labels.  This is used as the null distribution to assess statistical significance.  We use $\bfPP$, $\bfEP$, and $\VP$ to denote the probability, expectation, and variance under the permutation null.

\subsection{$\RC$}
\label{sec:rc0}

%    \subsubsection{Mean and Variance of $R_{ C_0}$ under the Permutation Null}
%  The following lemma states that the first two moments of $\RC$ under the permutation null can be computed instantaneously using basic summary statistics of the graph and cell counts of the contingency table.

\begin{lemma}\label{thm:EPVP}
The mean and variance of $R_{ C_0}$ under the permutation null are
  \begin{align}
    \label{eq:EP_RG0}
    \bfEP[R_{ C_0}] &= (N-K+|  C_0|)2p_1, \\
    \VP[R_{ C_0}] & = 4(p_1-p_2)(N-K + 2|  C_0| + \sum_u |\calE_u|^2/(4 m_u) - \sum_u |\calE_u|/m_u ) \label{eq:VP_RG0}\\
 & \quad \quad + (6p_2-4p_1)(K-\sum_u 1/m_u) + p_2\sum_{(u,v)\in   C_0} 1/(m_um_v) \nonumber \\
& \quad \quad + (N-K+|  C_0|)^2(p_2-4p_1^2), \nonumber
  \end{align}
where
\begin{equation}
  \label{eq:p1p2}
  p_1 = \frac{n_a n_b}{N(N-1)}, \quad p_2 =  \frac{4n_a (n_a -1) n_b (n_b -1)}{N(N-1)(N-2)(N-3)}.
\end{equation}

\end{lemma}

The proof of Lemma \ref{thm:EPVP} is given in Supplementary material \ref{sec:prooflemma}.

We need conditions to guarantee the convergence to normality of $R_{C_0}$ after standardization by its mean and variance.
\begin{condition}\label{condition:hub}
  $$\sum_u m_u(m_u +|\mathcal{E}_u^{C_0}|) (m_u + \sum_{v\in\mathcal{V}_u} m_v + |\mathcal{E}_{u,2}^{C_0}|) \sim o(K^{3/2}),$$
\vspace{-1em}
$$\sum_{(u,v)\in  C_0}(m_u + m_v + |\mathcal{E}_u^{C_0}| + |\mathcal{E}_v^{C_0}|) (m_u + m_v + \sum_{w\in (\mathcal{V}_u \cup \mathcal{V}_v)} m_w + |\mathcal{E}_{u,2}^{C_0}| + |\mathcal{E}_{v,2}^{C_0}|) \sim o(K^{3/2}).$$
\end{condition}
\noindent Condition \ref{condition:hub} constrains the size of ``hubs'' in the graph: The node degrees in $ C_0$ and the number of observations in each category must not be too large.  It can be simplified to stronger conditions that are easier to comprehend, for example the following.
\begin{customcnd}{$1''$}\label{condition:hub2}
  $\beta^6\lambda^2$ and $\lambda^8$ are both $o(K)$.
\end{customcnd}
%It follows easily from condition \ref{condition:hub} as $|\mathcal{E}_{u,2}^{C_0}|\leq \lambda^2$, and $| C_0|=\sum_u|\mathcal{E}_u^{C_0}|/2 \leq \lambda \beta/2$.
\noindent  The second condition is usually trivial:
\begin{condition}\label{condition:other}
  $N, ~|C_0|,$  and $\sum_{(u,v)\in  C_0}\frac{1}{m_um_v}$ are all  $\mathcal{O}(K)$.
\end{condition}

The asymptotic distribution of the standardized form of $R_{C_0}$ is given in the following theorem.
\begin{theorem}\label{thm:normality_P}
Assume that Conditions \ref{condition:hub} and \ref{condition:other} hold.  Under the permutation null the standardized statistic
$(R_{ C_0} - \bfEP[R_{ C_0}])/\sqrt{\VP[R_{ C_0}]}$ converges in distribution to $N(0,1)$ as $K\rightarrow \infty$ and $n_a/N$ is bounded away from 0 and 1.
\end{theorem}

The proof of Theorem \ref{thm:normality_P} is given in Supplementary material \ref{sec:proofthm}.

Theorem \ref{thm:normality_P} can be applied to any type of graph, allowing for repeated observations of each node.  Since the statistics in \cite{friedman1979multivariate} and \cite{rosenbaum2005exact} do not allow ties, their asymptotic normality results are also restricted to the case where each node is observed only once.  To compare Theorem \ref{thm:normality_P} to its counterparts, we let $G=C_0$ and assume that $m_u \equiv 1$.  Thus  $N=K$, and $\sum_{(u,v)\in C_0}\frac{1}{m_um_v} = |C_0| = |G|.$  Condition 2 requires that $|G|\sim \mathcal{O}(K)$ and Condition \ref{condition:hub} can be simplified to
$$\sum_u|\mathcal{E}_u^{G}| |\mathcal{E}_{u,2}^{G}| \sim o(K^{3/2}),$$
\vspace{-1.5em}
$$\sum_{(u,v)\in  G} (|\mathcal{E}_u^{G}| + |\mathcal{E}_v^{G}|) (|\mathcal{E}_{u,2}^{G}| + |\mathcal{E}_{v,2}^{G}|) \sim o(K^{3/2}).  $$

Theorem \ref{thm:normality_P} implies the asymptotic normality result in \cite{rosenbaum2005exact} since $|\mathcal{E}_u^{G}|\equiv 1, |\mathcal{E}_{u,2}^{G}|\equiv 1, |G|=K/2$ for MDP.  Friedman and Rafsky proved a more general condition for asymptotic normality of sums \eqref{sumedges} after standardization: For sparse graphs where $|G| \sim \mathcal{O}(K)$, the number of edge pairs that share a common node must be $\mathcal{O}(K)$.  Condition \ref{condition:hub} is neither stronger or weaker than Friedman and Rafsky's condition.  For example, a graph with one node of degree $K^{1/2}$ and all other nodes of degree 1 satisfies Friedman and Rafsky's condition but not Condition \ref{condition:hub}, since $\sum_{(u,v)\in  G}|\mathcal{E}_u^{G}||\mathcal{E}_{u,2}^{G}| = O(K^{3/2})$.  On the other hand, a graph with $\sqrt{K}$ nodes of degree $K^{0.3}$ and all other nodes of degree 1 satisfies Condition \ref{condition:hub} but not Friedman and Rafsky's condition.

\subsection{$T_{ C_0}$}
\label{sec:r_umst}
Here is the counterpart of Lemma \ref{thm:EPVP} for $R_{C_0}$.  Its proof is given in Supplementary material \ref{sec:proof-lemma-refl}.

\begin{lemma}\label{lemma:EPVP_TC0}
The mean and variance of $T_{ C_0}$ under the permutation null are
 \begin{align}
  \label{eq:EP_TC0}
  \bfEP[T_{ C_0}] & = (\sum_u m_u(m_u-1) + 2\sum_{(u,v)\in C_0} m_um_v ) p_1, \\
  \label{eq:VP_TC0}
  \VP[T_{ C_0}] & = (p_1-p_2) \sum_um_u(m_u + \sum_{v\in \mathcal{V}_u} m_v -1) (m_u + \sum_{v\in \mathcal{V}_u} m_v -2) \\
& \quad \quad + (p_1 - p_2/2)(\sum_u m_u(m_u-1) + 2\sum_{(u,v)\in C_0} m_um_v ) \nonumber \\
& \quad \quad + (p_2 - 4p_1^2) (\sum_u m_u(m_u-1) + 2\sum_{(u,v)\in C_0} m_um_v )^2, \nonumber
\end{align}

\vspace{-0.5em}

where $p_1$ and $p_2$ are given in \eqref{eq:p1p2}.

\end{lemma}

% \begin{lemma}\label{lemma:EBVB_TC0}
% The mean and variance of $R_\uMST$ under the bootstrap null are
% \begin{equation}
%   \label{eq:EP_TC0}
%   \bfEB[R_\uMST] = \left(\sum_u m_u(m_u-1) + 2\sum_{(u,v)\in \Mstar} m_um_v \right) p_3,
% \end{equation}
% \begin{align}
%   \label{eq:VP_TC0}
%   \VB[R_\uMST] & = (p_3-p_4) \sum_um_u(m_u + \sum_{v\in \mathcal{V}_u} m_v -1) (m_u + \sum_{v\in \mathcal{V}_u} m_v -2) \\
% & \quad \quad + (p_3 - p_4/2)\left(\sum_u m_u(m_u-1) + 2\sum_{(u,v)\in \Mstar} m_um_v \right) \nonumber
% \end{align}
% where $p_3$ and $p_4$ are given in \eqref{eq:p3}.
% \end{lemma}

The next theorem gives a sufficient condition for asymptotic normality of $T_{C_0}$ under the permutation null.
\begin{theorem}\label{thm:normality_P_TC0}
If $\sum_um_u(m_u + \sum_{v\in \mathcal{V}_u} m_v)^2 \sim \mathcal{O}(N)$, then under the permutation null distribution, the standardized statistic $(T_{ C_0} - \bfEP[T_{ C_0}])/\sqrt{\VP[T_{ C_0}]},$ where $\bfEP[T_{ C_0}]$ and $\VP[T_{ C_0}]$ are given in \eqref{eq:EP_TC0} and \eqref{eq:VP_TC0}, converges in distribution to $N(0,1)$ as $N\rightarrow \infty$, and $n_a/N$ bounded away from 0 and 1.
\end{theorem}

\begin{proof}
  Let $\overline{G}$ be the uMST on subjects.  Then as long as $\sum_{i=1}^N |\mathcal{E}_i^{\overline{G}}|(|\mathcal{E}_i^{\overline{G}}|-1) \sim \mathcal{O}(N)$, asymptotic normality can be ensured following \cite{friedman1979multivariate}'s result.  Notice that if $i$ is in category $u$, then $|\mathcal{E}_i^{\overline{G}}| = (m_u-1) +\sum_{\mathcal{V}_u} m_v$.
\end{proof}

\subsection{Checking the $P$-values Under Normal Approximations}
\label{sec:pvaluecheck}
We checked the normal approximations to the $p$-values of the three graph-based statistics -- $R_\CuMST$, $R_\CuNNG$ and $R_\uMST$ -- through simulation.  We adopted the setting of the haplotype example.  In each simulation run, $N$ haplotypes with length $l$ were generated uniformly from all possible haplotypes with length $l$.  They were assigned to the groups with equal probability.  For each simulation run, we calculated the difference between theoretical $p$-values from the normal approximation and the permutation $p$-values from 10,000 permutations for the three statistics.  We considered different sparsity settings by varying $l$, which controls the number of categories, and $N$.  Under each setting, 100 simulation runs were done, with the boxplots of the differences between theoretical and simulation $p$-value shown in Figure \ref{fig:pvalue}.  We increased $l$ from 6 to 10, and thus the number of possible categories considered is from 64 to 1024.  The sample size $N$ varies from 100 to 1000.  This spectrum of values is reasonable for a genetic association study.

Simulation results under this setting show that the normal approximation is better for $R_\CuMST$ and $R_\CuNNG$ than for $R_\uMST$.  Accuracy of normal approximations improved for all statistics as $l$ and $N$ increase.  For $R_\CuMST$ and $R_\CuNNG$, when the number of possible categories is larger than 256 and the number of observations larger than 200, the $p$-value from normal approximation was quite accurate.  For $R_\uMST$, the number of observations needs to be larger than 500 to achieve similar accuracy.  For $R_\uMST$, when the number of possible categories is larger than the number of observations, the $p$-value calculated from the normal approximation was negatively biased, and thus less conservative.  The bias is less severe for $R_\CuMST$ and $R_\CuNNG$, while still problematic when the number of possible categories is 1024 and the number of observation only 100.  Skewness correction can be done to make the theoretical $p$-values more accurate, but when $N$ is small, it would be easier to do permutation directly.

\begin{figure} 
  \centering
  \includegraphics[width=1.1\textwidth]{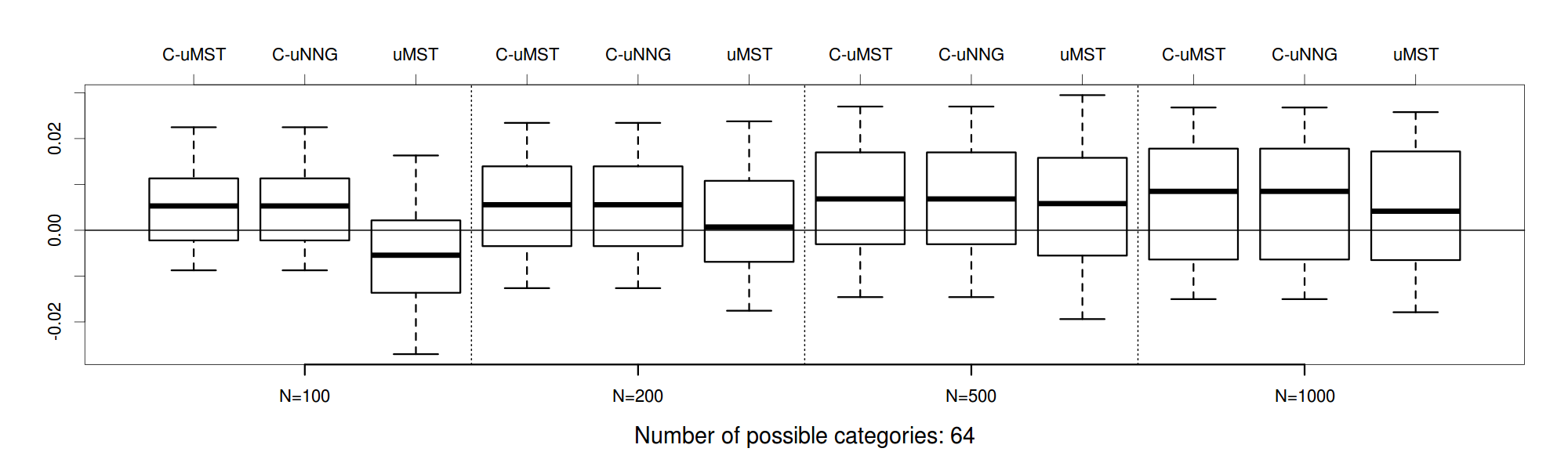}
  \includegraphics[width=1.1\textwidth]{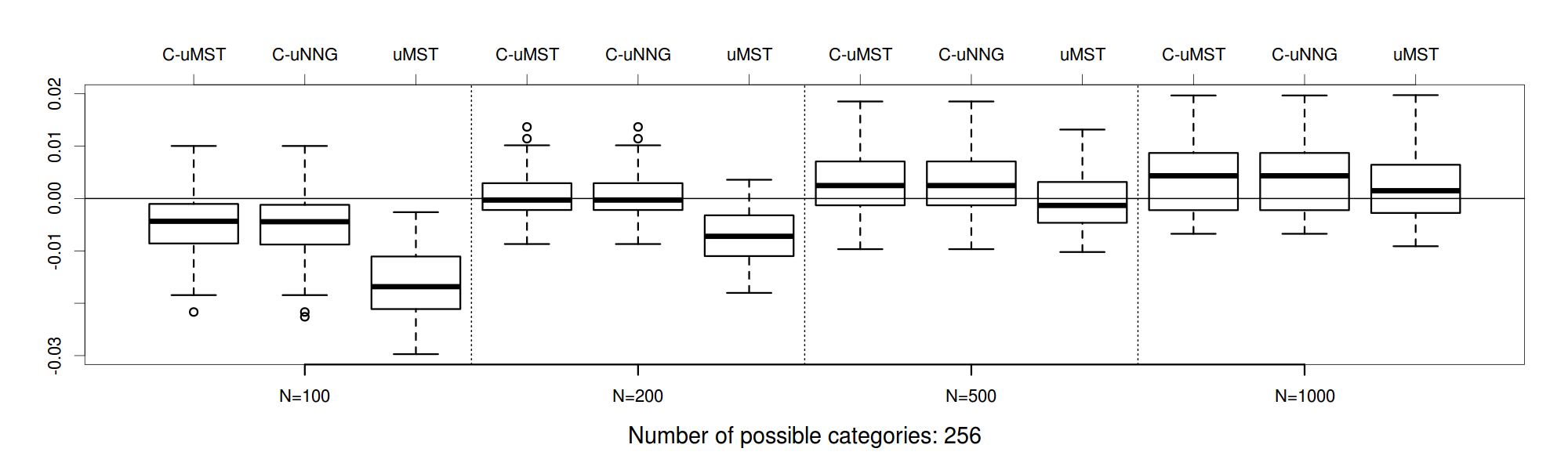}
  \includegraphics[width=1.1\textwidth]{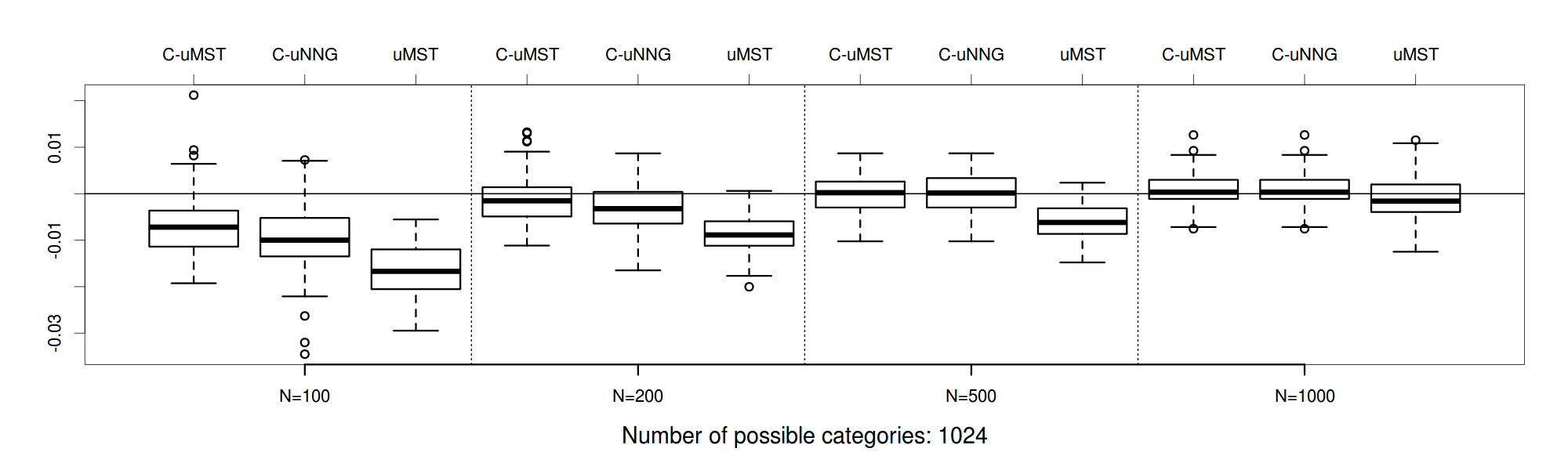}
  \caption{Boxplots for the differences between $p$-values calculated from normal approximation and 10,000 permutations.}
  \label{fig:pvalue}
\end{figure}

\section{Conclusions and Discussion}
\label{sec:conclusion}

Our approach to compare two categorical samples, useful when the contingency table is sparsely populated, utilizes a graphical encoding of the similarity between categories to improve the power of two-sample comparison.  Simulations and examples show that utilizing graphical information improves the power over the deviance and Pearson's Chi-square tests.  Proposed statistics are shown to be asymptotically normal after standardization under assumptions that limit the hub size and density of the graph.  This allows instantaneous type I error control for large data sets.

The power of the new approach depends on the choice of an informative similarity measure between categories and relies on domain knowledge that is specific to the application.  %For ranking data from surveys, one can start with the standard distance measures used in Example \ref{sec:ranking}.
When the number of categories is large, drawing relationships between categories is a necessary and often default step in analyzing the data.

%Both $R_\CuMST$ and $R_\uMST$  work well when the similarity information is effective with $R_\uMST$ usually having better power.  However, when the similarity measure is not as informative, $R_\uMST$ can have very low power, even when compared to Chi-square tests.  %In this sense, we would recommend $R_\CuMST$ (or $R_{ C_0}$ if a better $C_0$ is implied from a specific application) to be used since we could not tell whether most of the similarity measure we use are informative or not in real problems.  However, if we are sure that the similarity measure are effective, then $R_\uMST$ is recommended.
%In our simulation studies derived from the Haplotype problem, the normal approximation is more accurate for $R_\CuMST$ than for $R_\uMST$.  For $R_\uMST$, $p$-values obtained from normal approximation are lower than actual $p$-values in extremely sparse situations.  All $p$-value approximations work well when the sample size is comparable to the number of categories.

Generalization of our approach to multi-sample comparison is straightforward by letting $g_i$ take $K^\prime$ distinct values, where $K^\prime$ is the number of groups.

\appendix
% \section{Appendix}

\section{The Test Statistic Based on $R_\aMDP$}
\label{sec:Rmdp}

We assume $N$, the total number of observations, to be even.  Let $K_0$ be the number of categories containing an odd number of subjects.  Since $N$ is even, $K_0$ is even. ($K_0$ can be 0.).  Without loss of generality, let categories $1,\dots,K_0$ be the categories containing an odd number of subjects, and categories $K_0+1,\dots, K$ be the categories containing an even number of subjects.  More notations are as follows.

\begin{itemize}
\item $\calA = \{ \bx = (x_1, \dots, x_{K_0})^T:
  x_i\in\{a,b\}, i=1,\dots,K_0\}$: all possible
  combinations of group identities of the subjects with one from each of the categories containing an odd number of subjects.
\item $R_0(n_a, n_b)$: the number of edges connecting subjects from different
  groups averaged over all perfect pairings of $n_a$ points from group $a$ and $n_b$
  points from group $b$ in the same category, with $n_a + n_b$ being even.
\item $R_\bx, \bx\in \calA$: the number of edges connecting subjects from
   different groups averaged over all MDPs on categories $1,\dots,K_0$.
\end{itemize}

\begin{assumption}
  \label{assump:cm} If a category has an even number of subjects, the
subjects are paired within the category.
\end{assumption}

Assumption \ref{assump:cm} is usually true for MDP on subjects for categorical data.  It is explicitly stated here to avoid complications when the triangle inequality becomes equality in the distance metric for any three categories.

\begin{proposition}
Under Assumption \ref{assump:cm}, the test
statistic based on averaging \eqref{sumedges} over all MDPs is:
\begin{align}
  \label{eq:Rmdp}  R_\aMDP
&= \sum_{k = K_0 + 1}^K R_0(n_{ak}, n_{bk}) +
\frac{1}{\prod_{k=1}^{K_0} m_k} \sum_{\bx\in\calA} \left\{
\prod_{i=1}^{K_0}n_{x_i i} \left[ R_\bx +
\sum_{j=1}^{K_0} R_0(n_{x_j j} -1, n_{x_j^c j})
\right] \right\},
\end{align} where $x_i^c = \left\{ \begin{array} {ll} b &
\text{if } x_i = a \\ a & \text{if } x_i = b \end{array}
\right.$,
\begin{align}
\label{eq:R0}
  R_0(n_a, n_b) & = \sum_{i\in \mathcal{S}}i \left(\begin{array}{c}n_a\\i\end{array}\right)  \left(\begin{array}{c}n_b\\i\end{array}\right)i!\ (n_a-i-1)!!\ (n_b-i-1)!!/ (n_a+n_b-1)!!
\end{align}
with
\begin{align}
\mathcal{S} & = \left\{\begin{array}{ll} \{0,2,\dots, n_a\wedge n_b\} & \text{ if } n_a \text{ and } n_b \text{ both even} \\ \{1,3,\dots, n_a\wedge n_b\} & \text{ if } n_a \text{ and } n_b \text{ both odd}\end{array} \right. , \nonumber\\
  \label{eq:Rx}
  R_\bx & = |\Omega^*|^{-1}\sum_{\omega^* \in \Omega^*} \sum_{(i,j)\in \omega^*} I_{x_i\neq x_j},
\end{align}
where $\omega^*$ is an MDP on categories $1,\dots,K_0$, and $\Omega^*$ is the set of all these $\omega^*$'s.
\end{proposition}

\begin{proof}
  Consider the simpler case of one category with $n_a$ subjects from group $a$ and $n_b$ subjects from group $b$, with $n_a+n_b$ even.  Since all subjects are in the same category, any perfect pairing is an MDP. There are in total
$(n_a+n_b-1)!!$ different perfect pairings.

When both $n_a$ and $n_b$ are even, the possible numbers of edges connecting different groups are $0, 2, \dots, n_a \wedge n_b$.  Among all the $(n_a+n_b-1)!!$ perfect pairings, the
number of perfect pairings having $i\in \{0, 2, \dots, n_a \wedge n_b\}$
edges connecting different groups is
\begin{equation}
  \label{eq:1}
  \left(\begin{array}{c}n_a\\i\end{array}\right)  \left(\begin{array}{c}n_b\\i\end{array}\right)i!\ (n_a-i-1)!!\ (n_b-i-1)!!.
\end{equation}
When both $n_a$ and $n_b$ are odd, the possible numbers of edges connecting different groups are $1, 3, \dots, n_a \wedge n_b$.  Among all the $(n_a+n_b-1)!!$ perfect pairings, the
number of perfect pairings having $i\in \{1, 3, \dots, n_a \wedge n_b\}$
edges connecting different groups also has the form \eqref{eq:1}. \eqref{eq:R0} follows immediately.

Under Assumption \ref{assump:cm}, an MDP on all subjects would be an MDP on categories $1,\dots,K_0$, ($\omega^*$), embedded on the subjects similar to the MST case, with all other subjects paired within each category, so \eqref{eq:Rmdp} follows naturally.

\end{proof}
\begin{remark}
  \emph{If $N$, the total number of observations, is odd, we can add a pseudo category with one subject, whose distance to any other category is 0.  All derivations are the same, except that the edge containing the pseudo category is discarded from the MDP on categories in later steps.}
\end{remark}

\section{Computation Time for $R_\aMST$ and $R_\uMST$}
\label{sec:comptime}

The main tasks for computing $R_\aMST$ and $R_\uMST$ are to enumerate all MSTs on categories for $R_\aMST$ and to list the edges in $\Mstar$ for $R_\uMST$.  Other tasks can be finished in $\mathcal{O}(K)$ time.

% To be more general, we start with a graph on the categories, which we call $G$, that has $K$ nodes, and an edge between any two nodes with finite distance.  we let $E$ be the number of edges in $G$.  If we start with the distance matrix on the $K$ categories, then $G$ is the complete graph and $E=K(K-1)/2$.

Let $G$ be the complete graph on $K$ categories, $|G| = K(K-1)/2$. \cite{eppstein1995representing} proposed a graph operation called the sliding transformation which, when applied to $G$, produces an equivalent graph such that the MSTs on categories correspond one-for-one with the spanning trees of the equivalent graph.  The enumeration of all spanning trees, without having to optimize for total distance, is relatively straightforward.  Thus, we adopted the following computational approach.  Use Eppstein's method to construct the equivalent graph of $G$, enumerate all spanning trees of the equivalent graph, then transform back to get the set of MSTs on $G$.  The sliding transformation constructs the equivalent graph in  $\calO(|G|+ K\log K) = \calO(K^2)$ time.  To perform the sliding transformation, an initial MST is needed.  Prim's algorithm can be used to obtain the initial MST, which requires $\calO(K^2)$ time, not increasing the time complexity.  The theoretical justification of this algorithm can be found in \cite{eppstein1995representing} and \cite{chen2013graph}, which completes many of the proofs of \cite{eppstein1995representing}.

After removing any loops formed during the the sliding transformations, each remaining edge appears in at least one spanning tree of the equivalent graph, thus appearing in at least one MST on $G$.  Now we have the list of edges in uMST on $G$, and thus $R_\uMST$ can be calculated in $\mathcal{O}(K^2)$ time.

For enumerating all spanning trees of the equivalent graph, the algorithm proposed by \cite{shioura1995efficiently} is used; it requires $\mathcal{O}(K+|G|+M) = \mathcal{O}(K^2+M)$ computation time, proven to be optimal in time complexity.    Shioura and Akihisa's algorithm starts from a spanning tree formed by depth-first search, then replaces one edge at a time using cycle structures in the graph, traversing the space of all spanning trees of the graph.  Hence, computing $R_\aMST$ takes $\mathcal{O}(K^2+M)$ time.

%%%%%%%%%%%%%%%%%%%%%%%%%%%%%%%%%%%%%%%%%%%%%%%%%%%%%%%%%%%%%%%%%%%%%%%%%%%%%%%%%%%%%%%%%%%%%%%%%%%%%%%%%%%%%%%%%%%%%%%%%%%%

\noindent {\large\bf Acknowledgments}

Hao Chen was supported by NSF DMS Grant 1043204 and an NIH Training Grant.  Nancy R. Zhang was supported by NSF DMS Grant 0906394 and NIH Grant R01 HG006137-01.  We thank one of the reviewers for bringing Critchlow (1985)'s work to our attention. 

\par

%%%%%%%%%%%%%%%%%%%%%%%%%%%%%%%%%%%%%%%%%%%%%%%%%%%%%%%%%%%%%%%%%%%%%%%%%%%%%%%%%%%%%%%%%%%%%%%%%%%%%%%%%%%%%%%%%%%%%%%%%%%%

\bibliographystyle{apa}
\bibliography{cat_comp}

%%%%%%%%%%%%%%%%%%%%%%%%%%%%%%%%%%%%%%%%%%%%%%%%%%%%%%%%%%%%%%%%%%%%%%%%%%%%%%%%%%%%%%%%%%%%%%%%%%%%%%%%%%%%%%%%%%%%%%%%%%%%

\vskip .65cm
\noindent
Department of Statistics, Stanford University
\vskip 2pt
\noindent
E-mail: haochen@stanford.edu
\vskip 2pt

\noindent
Department of Statistics, The Wharton School, University of Pennsylvania
\vskip 2pt
\noindent
E-mail: nzh@wharton.upenn.edu
\vskip .3cm

%%%%%%%%%%%%%%%%%%%%%%%%%%%%%%%%%%%%%%%%%%%%%%%%%%%%%%%%%%%%%%%%%%%%%%%%%%%%%%%%%%%%%%%%%%%%%%%%%%%%%%%%%%%%%%%%%%%%%%%%%%%%
%%%%%%%%%%%%%%%%%%%%%%%%%%%%%%%%%%%%%%%%%%%%%%%%%%%%%%%%%%%%%%%%%%%%%%%%%%%%%%%%%%%%%%%%%%%%%%%%%%%%%%%%%%%%%%%%%%%%%%%%%%%%

\end{document}